%% file: main.tex
\documentclass[runningheads]{llncs}

\usepackage[T1]{fontenc}
\usepackage[utf8]{inputenc}

\usepackage{
	algorithm,
	algorithmic,
  amsmath,
  amssymb,
  booktabs,
  xurl,
  color,
  xcolor,
  graphicx, 
  multirow,
  listings,
  stmaryrd,
  tabularx,
  capt-of
}

\usepackage[hidelinks]{hyperref} 
\usepackage[nameinlink]{cleveref}
\creflabelformat{equation}{#2\textup{#1}#3} 

\crefname{section}{Sect.}{sections}
\Crefname{section}{Section}{Sections}
\crefname{table}{Tbl.}{tables}
\Crefname{table}{Table}{Tables}
\crefname{proposition}{Prop.}{propositions}
\Crefname{proposition}{Proposition}{Propositions}

\renewcommand{\figurename}{Figure}

\usepackage [
  url=false,
  doi=false,
  isbn=false,
  maxbibnames=99,
  backend=biber
] {biblatex}

\let\vec\mathbf

\title{Symbolic Quantitative Information Flow for Probabilistic Programs}

\author{Philipp Schröer\inst{1}\orcidID{0000-0002-4329-530X} \and Francesca Randone\inst{2}\orcidID{0009-0002-3489-9600} \and Raúl Pardo\inst{3}\orcidID{0000-0003-0003-7295} \and Andrzej Wąsowski\inst{3}\orcidID{0000-0003-0532-2685}}

\institute{RWTH Aachen University, Germany \email{phisch@cs.rwth-aachen.de}
\and
University of Trieste, Italy \email{frarandone@gmail.com}
\and
IT University of Copenhagen, Denmark,
\email{\{raup,wasowski\}@itu.dk}}

\input{macros.tex}

\titlerunning{Symbolic Quantitative Information Flow for Probabilistic Programs}
\authorrunning{P.\ Schröer, F.\ Randone, R.\ Pardo, A.\ Wąsowski}

\begin{document}

\maketitle

\input{abstract.tex}

\input{intro.tex}
\input{related_work.tex}
\input{programming_language}
\input{metrics}
\input{case_studies}

\input{acknowledgements}

\printbibliography

\end{document}

%% file: macros.tex
\newcommand{\morespace}[1]{\quad#1\quad}

\newcommand{\normal}{\ensuremath{\mathcal{N}}}
\newcommand{\bernoulli}{\ensuremath{\mathcal{B}}}

\newcommand{\rvx}{\ensuremath{X}}
\newcommand{\rvy}{\ensuremath{Y}}
\newcommand{\income}{\ensuremath{\mathit{inc}}}
\newcommand{\female}{\ensuremath{\mathit{fem}}}
\newcommand{\nfem}{\ensuremath{\mathit{nfem}}}
\newcommand{\out}{\ensuremath{\mathit{output}}}
\newcommand{\response}{\ensuremath{r}}
\newcommand{\trv}[1]{#1} 
\newcommand{\var}{\ensuremath{\mathit{var}}}
\newcommand{\rvi}[0]{\ensuremath{I}}
\newcommand{\rvo}[0]{\ensuremath{O}}
\newcommand{\supportx}{\ensuremath{\mathcal{X}}}

\newcommand{\inputs}[0]{\ensuremath{\mathcal{I}}}
\newcommand{\outputs}[0]{\ensuremath{\mathcal{O}}}
\newcommand{\dists}[0]{\ensuremath{\mathcal{D}}}
\newcommand{\entropy}{\ensuremath{H}}
\newcommand{\program}[0]{\ensuremath{\pi}}
\newcommand{\kl}{\ensuremath{D_{\mathrm{KL}}}}
\newcommand{\doublemid}{~||~}
\newcommand{\density}[1]{f_{#1}}
\newcommand{\symSkip}{\texttt{skip}}
\newcommand{\symSemi}{;~}
\newcommand{\symAsgn}{\texttt{:=}}
\newcommand{\symDist}{\ensuremath {\sim}}
\newcommand{\symIf}{\texttt{if}}
\newcommand{\symElse}{\texttt{else}}

\newcommand{\symFor}{\texttt{for}}
\newcommand{\symObserve}{\texttt{observe}}
\newcommand{\blockStart}{\{}
\newcommand{\blockEnd}{\}}

\newcommand{\stmtSkip}{\symSkip}
\newcommand{\stmtSeq}[2]{{#1}\symSemi{#2}}
\newcommand{\stmtAsgn}[2]{{#1}~\symAsgn~{#2}}
\newcommand{\stmtDist}[2]{{#1}~\symDist~{#2}}

\newcommand{\stmtIf}[3]{\symIf~{#1}~\blockStart~{#2}~\blockEnd~\symElse~\blockStart~{#3}~\blockEnd}

\newcommand{\stmtFor}[3]{\symFor~{#1} \texttt{ in 1..}{#2} ~\blockStart~{#3}~\blockEnd}
\newcommand{\stmtObserve}[1]{\symObserve~{#1}}
\newcommand{\stmt}{S}
\newcommand{\bexp}{b}
\newcommand{\aexp}{e}

\newcommand{\wpStates}{\mathsf{States}}
\newcommand{\Expectations}{\mathbb{E}}
\newcommand{\CExpectations}{\Expectations \times \Expectations}
\newcommand{\wpState}{\sigma}
\newcommand{\wpVals}{{\mathbb{Q}_{\geq 0}}}
\newcommand{\expLeq}{\sqsubseteq}
\newcommand{\expA}{X}
\newcommand{\expB}{Y}
\newcommand{\expOne}{\mathbf{1}}
\newcommand{\sem}[1]{\llbracket{#1}\rrbracket}
\newcommand{\symWp}{\mathsf{wp}}
\renewcommand{\wp}[1]{\symWp\sem{#1}}
\newcommand{\symCwp}{\mathsf{cwp}}
\newcommand{\cwp}[1]{\symCwp\sem{#1}}

\newcommand{\substBy}[2]{[{#1} \mapsto {#2}]}
\newcommand{\iverson}[1]{[{#1}]}
\newcommand{\symTrue}{\texttt{true}}
\newcommand{\symFalse}{\texttt{false}}
\newcommand{\aux}[1]{\underline{{#1}}}

\newcommand{\condforceStart}{\llbracket}
\newcommand{\condforceEnd}{\rrbracket}
\newcommand{\condforce}[1]{\llbracket{#1}\rrbracket}

\newcommand{\wpcompcolor}[1]{\textcolor{lightgray!40!black}{#1}}
\newcommand{\wpcomp}[1]{{\wpcompcolor{\!\!{\fatslash}\!\!{\fatslash}~~\vphantom{G'} {#1}}}}

\newcommand{\probCondTerm}[4]{p({#2} \models_{#1} \diamond ({#3} = {#4}) \mid \neg\lightning)}

\newcommand{\gms}{\mathcal{G}}
\newcommand{\gm}{G}
\newcommand{\soga}[1]{\mathsf{SOGA}\llbracket{#1}\rrbracket}
\newcommand{\exact}[1]{\mathsf{Exact}\llbracket{#1}\rrbracket}
\newcommand{\mmop}{\mathcal{T}_\textrm{gm}}
\newcommand{\marg}[2]{\textrm{Marg}_{#1}(#2)}
\newcommand{\prob}[2]{P_{#1}(#2)}
\newcommand{\cond}[2]{#1_{\mid #2}}

\newcommand{\alog}{{-}\log}
\newcommand{\palog}{\log}


%% file: abstract.tex
\begin{abstract}
  It is of utmost importance to ensure that modern data intensive systems do not leak sensitive information. In this paper, the authors, who met thanks to Joost-Pieter Katoen, discuss symbolic methods to compute information-theoretic measures of leakage: entropy, conditional entropy, Kullback–Leibler divergence, and mutual information. We build on two semantic frameworks for symbolic execution of probabilistic programs. For discrete programs, we use weakest pre-expectation calculus to compute exact symbolic expressions for the leakage measures.  Using Second Order Gaussian Approximation (SOGA), we handle programs that combine discrete and continuous distributions. However, in the SOGA setting, we approximate the exact semantics using Gaussian mixtures and compute bounds for the measures. We demonstrate the use of our methods in two widely used mechanisms to ensure differential privacy: randomized response and the Gaussian mechanism.
\end{abstract}


%% file: intro.tex

\section{Introduction}
\label{sec:intro}

\emph{Data is the modern oil.}\footnote{\url{https://www.economist.com/leaders/2017/05/06/the-worlds-most-valuable-resource-is-no-longer-oil-but-data}} The data processing systems are pipes; pipes that leak. They leak not oil, but sensitive, often personal and private, information. Authorities around the world have developed data protection regulations and laws to ensure this leakage is limited: HIPAA (1996) in the US~\cite{hipaa1996}, the GDPR in the EU/EEA (2016)~\cite{gdpr} or the Cybersecurity Law of the People's Republic of China (2016)~\cite{noauthor_cybersecurity_2016}. Whenever data is being processed, and the results are being released, the data processors now need to ask whether releasing the computation results does not leak too much sensitive information.  This has led to establishment of \emph{quantitative information flow} as an active area of research~\cite{DBLP:conf/csfw/ClarksonMS05,DBLP:conf/fossacs/Smith09,qifbook.2020}.

The key idea of quantitative information flow is to analyze a computer program as if it was a communication channel \cite{ShannonWeaver49} between sensitive secret inputs and publicly released outputs.  A channel models an attacker that has some prior information about the input and that can observe the output.  Observing the output may allow the attacker to learn something about the input.  Some learning typically happens, as otherwise, either the attacker must have known the sensitive input to begin with, or the program output is not useful---it is not correlated with the input, e.g.\ constant or random.  Thus the question is not really whether a leakage happens, as in qualitative information flow, but rather how much information is leaked about each secret and under what circumstances.  To this end, researchers have proposed a number of metrics. In this paper, we focus on a selection of popular ones: Shannon entropy, Kullback-Leibler divergence, and mutual information.

Many methods have been proposed for computing quantitative information flow~\cite{qifbook.2020,DBLP:conf/cav/ChothiaKN13,DBLP:conf/esorics/ChothiaKN14,QUAIL,HyLeak,spire,DBLP:conf/sp/Cherubin0P19,romanelli.leaves.2020}.  Some of the authors of the present paper, have been involved in the development of
Privug\,\cite{privug}, a recent addition that relies on Bayesian inference (probabilistic programming) to quantify privacy risks in data analytics programs. In Privug, the attacker's knowledge is modeled as a probability distribution over secret inputs, and it is subsequently conditioned on the disclosed program outputs. Bayesian probabilistic programming (sampling-based \cite{mcmc,privug} and static \cite{RPWeebiprq23}) is used to compute the posterior attacker knowledge after observing the outputs in the program.
An advantage of Privug is that it works on the program source code, and can be extended to compute most information leakage metrics~\cite{qifbook.2020}.  No actual data is needed for the analysis.

So far, the static methods used in Privug have not been extended to compute leakage metrics (only methods for programs under strict restrictions or sampling-based methods have been explored). In this work, we investigate the theoretical underpinning of approximating leakage metrics for probabilistic programs symbolically using two foundations of semantics for two classes of probabilistic programs: (1) the weakest-pre-expectation (WPE) semantics  for discrete programs \cite{10.1145/3622870,kaminski_phd,DBLP:series/mcs/McIverM05}, and (2) the Gaussian mixture semantics  for continuous programs \cite{RandoneBIT24}. In particular, we contribute:%
\begin {itemize}

  \item Definitions of WPE and Gaussian-Mixture semantics for a simple probabilistic programming language;

  \item Methods for symbolic calculation (or approximation) of entropy, KL-diver\-gence, and mutual information for discrete and continuous probabilistic program, under suitable assumptions;

  \item Examples of analysis of differential privacy mechanisms using the provided symbolic calculation methods.

\end {itemize}%
For years, Joost-Pieter Katoen has led a very successful research group on probabilistic models and programs at RWTH, a very influential environment and a source of many inspiring ideas.  In particular, the work on symbolic analyses in the WPE semantics \cite{10.1145/3622870,kaminski_phd} and on using generating functions \cite{DBLP:conf/cav/ChenKKW22} brought Francesca, Philipp, and Andrzej together in Aachen in Spring of 2023. If not for Joost-Pieter, we would not have known each other, and this work would not have existed.  Likewise, this work would not have been possible without the methods developed under Joost-Pieter's supervision. Dear Joost-Pieter, thank you for providing these foundations and for helping us to meet. Please accept this paper as a gift.  Happy Birthday!


%% file: related_work.tex
\section{Related Work}
\label{sec:related_work}

Clarkson and coauthors introduced the attacker belief revision framework about two decades ago~\cite{DBLP:conf/csfw/ClarksonMS05}. Already there, they have used KL divergence as a measure of information leakage. They introduce the notion of experiment to model the attacker and system behavior; to quantify leakage in different situations. Their framework defines quantitative integrity metrics, namely, contamination, channel suppression and program suppression~\cite{DBLP:journals/mscs/ClarksonS15}---these metrics are entropy-based measures. They show the relation of these metrics to standard notions of privacy for databases, e.g., $k$-anonymity~\cite{sweeney2002k}, $l$-diversity~\cite{li2006t} and differential privacy~\cite{DBLP:journals/fttcs/DworkR14}. Furthermore, they extend the framework to quantify leakage as the system and attacker behavior changes~\cite{DBLP:conf/sp/MardzielAHC14}. To this end, they update the notion of experiment with strategies that model the system and attacker behavior. The notion of experiment and strategy is encoded as probabilistic programs with bounded loops. Thus, they can be captured in the programming language presented in this paper. Although these works consider similar leakage measures as in this paper, they only consider systems with discrete random variables.

After the seminal work of Smith~\cite{DBLP:conf/fossacs/Smith09}, quantitative information flow metrics based on Bayes vulnerability have become more prominent in the quantitative information flow literature~\cite{qifbook.2020}.
Bayes vulnerability focuses on computing the expected reward (or probability of success) of an attack.
These metrics have been applied to analyze census data, the topics API (an alternative to cookies proposed by Google for interest based advertising\footnote{\url{https://github.com/patcg-individual-drafts/topics}}), machine learning systems, voting protocols, or differential privacy~\cite{DBLP:journals/popets/AlvimFMMN22,DBLP:conf/wpes/AlvimFMN23,DBLP:journals/ijon/GrossoPPP23,qifbook.2020}.
Unlike our work, these applications of quantitative information flow do not focus on program analysis, but on system analysis.
They (manually) encode the target systems as channels (in the information theoretic sense), and compute leakage metrics using the channel encoding.
Instead, we provide a probabilistic semantics for programs, and focus on quantifying leakage symbolically in those programs directly, not using matrix representations.
For systems using discrete random variables, we could compute Bayes vulnerability metrics by directly applying the definition---although this would be inefficient for systems with large output domains.
We leave as future work devising methods to efficiently compute Bayes vulnerability metrics (and bounds) using the probabilistic semantics presented in this paper.

There exist some works that attempt to automate quantitative information flow analysis, either by estimating leakage via sampling or by computing exact expressions.
Leakage estimation works can be either black-box or white-box techniques.
There are black-box methods to compute Bayes vulnerability~\cite{DBLP:conf/sp/Cherubin0P19,DBLP:conf/ccs/0002CPP20}, and for entropy-based measures such as mutual information or min-entropy~\cite{DBLP:conf/csfw/ChothiaG11,DBLP:conf/cav/ChothiaKN13,DBLP:conf/esorics/ChothiaKN14}.
Privug is a white-box method that uses Markov Chain Monte Carlo to estimate leakage measures \cite{privug}.
QUAIL~\cite{QUAIL} (also white-box) uses forward state exploration to build a Markov chain, which is then used to estimate mutual information.
In the present paper, we instead compute exact solutions and upper bounds to support the worst-case analysis.
Crucially, underestimating leakage could lead to missing vulnerability in the target systems.
In a parallel work \cite{RPWeebiprq23}, we have used multivariate Gaussian distributions as a semantics to compute exact mutual information and KL divergence.
SPIRE~\cite{spire} uses the exact inference engine PSI~\cite{PSI,gehr_psi_2020} to compute leakage metrics.
Although we do not explore tools in the present paper, the two semantics we use are supported by tools~\cite{RandoneBIT24,10.1145/3622870}, and the Gaussian semantics engine has been shown to be more efficient than PSI for programs using only continuous random variables~\cite{RRPWqprgm24}.
We leave as future work exploring the development of an automatic method to compute leakage metrics based on the findings in this work.

A line of work uses model-counting to compute entropy-based leakage measures~\cite{DBLP:conf/sp/BackesKR09,DBLP:conf/cav/Tizpaz-NiariC019,DBLP:conf/eurosp/BangRB18}.
Backes et al.\ propose to compute equivalence classes for secrets that produce the same output\,\cite{DBLP:conf/sp/BackesKR09}.
A similar method is proposed by Tizpas-Niari and coauthors to compute entropy-based measures\ \cite{DBLP:conf/cav/Tizpaz-NiariC019}, and use them as a basis to introduce synthetic delays in programs so that leakage does not go over a given threshold.
Bang et al.\ use symbolic execution and model-counting to estimate mutual information by sampling satisfying models from the path conditions computed by the symbolic executor~\cite{DBLP:conf/eurosp/BangRB18}.
These model-counting approaches implicitly assume a discrete uniform distribution over secrets.
In our work, we support multiple types of discrete and continuous distributions, which allows for modeling different types of attackers.
\looseness -1


%% file: programming_language.tex
\section{Programming Language}

Consider a simple programming language:
\begin {align}
  \stmt ::= ~ & \stmtSkip \mid \stmtSeq{\stmt}{\stmt} \mid \stmtAsgn{x}{\aexp} \nonumber
  \\
   & \mid \stmtDist{x}{\dists(\overline\theta)} \nonumber \\
   & \mid \symObserve(\bexp) \nonumber \\
   & \mid \stmtIf{\bexp}{\stmt}{\stmt} \nonumber \\
   & \mid \stmtFor{i}{n}{\stmt}
\end {align}

\noindent
The language admits random variables $x$ and expressions over them ($e$).  We leave the language of expressions underspecified for simplicity.
Distributions $\dists(\bar{\theta})$ may be discrete or continuous.
The symbolic semantics of the language depends on this choice: discrete (\Cref{sec:discrete_semantics}) or continuous (\Cref{sec:continouous_semantics}).
We also admit sequencing, if-statements and a limited version of for-loops (syntactic sugar for unrolling the body $n$ times, for a constant $n$).  For this overarching syntax, we now define discrete and continuous probabilistic semantics.

The developments below follow a parallel structure: the semantics are defined for discrete and continuous programs in separate subsections that share the same overarching syntax but are otherwise unrelated. In \Cref{sec:metrics}, we follow a similar parallel split, computing leakage metrics in the two semantics---the metrics are the same in both subsections, but the methods are different and the parallel developments are not directly related to each other.  Finally, in \Cref{sec:case_studies}, we present example analyses of suitable differential privacy mechanisms for the discrete and continuous models, again in parallel, to illustrate the corresponding methods.  We hope that this structure, will allow the reader to appreciate how the same tasks are solved differently in the two semantics for discrete and continuous programs.

\subsection{Discrete Symbolic Semantics: WPE}
\label{sec:discrete_semantics}

A \emph{(program) state} $\wpState \colon \mathbf{Vars} \to \mathbb{Q}_{\geq 0}$ maps every variable to a nonnegative rational value. We denote the set of program states by $\wpStates$. An \emph{expectation} maps every program state $\sigma$ to a nonnegative real value or infinity. The complete lattice of expectations $(\Expectations, \expLeq)$ is given by
\begin{multline}
    \Expectations = \{ \expA \mid X \colon \wpStates \to \mathbb{R}_{\geq 0}^\infty \}
    \enspace,\\
    \text{ where }
    \quad
    \expA \expLeq \expB \text{ iff } \text{for all } \wpState \in \wpStates: \expA(\wpState) \leq \expB(\wpState)~.
\end{multline}
A well-known expectation is the \emph{Iverson bracket} $\iverson{\bexp} \in \Expectations$ where $\bexp$ is a Boolean expression. For a state $\wpState$, $\iverson{\bexp}(\wpState) = 1$ iff $\bexp(\wpState)$ holds and $\iverson{\bexp}(\wpState) = 0$ otherwise. We write $\expA\substBy{x}{e}$ for the substitution (overriding) of $x$ by $e$ in the expectation $\expA$, that is: $\lambda \wpState.~ \expA(\wpState\substBy{x}{e(\wpState)})$, where $\wpState\substBy{x}{e(\wpState)}$ denotes the state $\wpState$ with $x$ overriden to the result of evaluating the expression $e$ in $\wpState$. We denote by $\expOne$ the expectation returning 1 for any $\wpState$.
Furthermore, given a one-bounded expectation $\expA \expLeq \expOne$, we will denote by $\alog(\expA)$ the expectation which, for every state $\wpState$, maps the result of $\expA(\wpState)$ to $\alog(\expA(\wpState))$.

\renewcommand \arraystretch {1.40}

\begin{table}[t]
    \centering
    \newcommand{\recwp}[1]{\wp{#1}}
    \begin{tabular}{l>{\hspace{9mm}}l}
        \toprule
        $\stmt$ & $\wp{\stmt}(\expA)$ \\
        \midrule
        $\stmtSkip$ & $\expA$ \\
        $\stmtAsgn{x_i}{\aexp}$ & $\expA\substBy{x_i}{\aexp}$ \\
        $\stmtDist{x_i}{\mu}$ & $\lambda \wpState.~ \int_\wpVals (\lambda v.~ X\substBy{x_i}{v})\, d\mu_\wpState$ \\
        $\stmtObserve{b}$ & $\iverson{b} \cdot \expA$ \\
        $\stmtSeq{\stmt_1}{\stmt_2}$ & $\recwp{\stmt_1}(\recwp{\stmt_2}(\expA))$ \\
        $\stmtIf{\bexp}{\stmt_1}{\stmt_2}$ & $\iverson{\bexp} \cdot \recwp{\stmt_1}(\expA) + \iverson{\neg b} \cdot \recwp{\stmt_2}(\expA)$ \\
        $\stmtFor{i}{n}{\stmt}$ & $\recwp{\stmtAsgn{i}{1}}(\recwp{\stmtSeq{\stmt}{\stmtAsgn{i}{i + 1}}}^n(\expA))$ \\
        \bottomrule
    \end{tabular}

    \bigskip

    \caption{The Weakest pre-expectation ($\symWp$) semantics of a statement $\stmt$ for an expectation $\expA \! \in \! \Expectations$}
    \label{table:wp-semantics}

\end{table}

For a program $\stmt$, the \emph{weakest pre-expectation transformer} $\wp{\stmt} \colon \Expectations \to \Expectations$ is a mapping between expectations. The expression $\wp{\stmt}(\expA)(\wpState)$ maps an expectation $\expA \in \Expectations$ and an initial state $\wpState$ to the expected value $\wp{\stmt}(\expA)(\wpState)$ of $\expA$ on termination of $S$. The $\symWp$ transformer is defined inductively in \Cref{table:wp-semantics}.

The $\symWp$ transformer computes expected values on termination, mapping paths where observations fail to zero.
The expected value $\wp{\stmt}$ of an expectation $\expA$ on termination of program $\stmt$ is thus given as follows.
A $\stmtSkip$ statement does nothing and returns $\expA$ unchanged.
The deterministic assignment $\stmtAsgn{x_i}{\aexp}$ results in a substitution of the variable $x_i$ in $\expA$ by the right-hand side $\aexp$.
Sampling $x_i$ from a distribution expression $\mu$ in state $\wpState$ is given by the Lebesgue integral of substitutions of $x_i$ by values from the distribution evaluated in the current state, $\mu_\wpState$.
The $\symWp$ semantics of $\stmtObserve{b}$ evaluates to zero if $b$ does not evaluate to $\symTrue$ in the current state $\wpState$, otherwise it is $X$.
The weakest pre-expectation transformer operates backwards, therefore the sequential execution first computes $\wp{\stmt_2}(\expA)$ and then applies $\wp{\stmt_1}$ to the result.
An \texttt{if} statement on the Boolean expression $b$ evaluates to $\wp{\stmt_1}(X)$ if $b$ evaluates to \symTrue{} and to $\wp{\stmt_2}(X)$ otherwise.
A $\symFor$ loop with exactly $n$ iterations results in an unfolding of $n$ loop iterations with counter increments.

To compute \emph{conditional} expected values, i.e. an expected values conditioned on observations not failing, we define the \emph{conditional weakest pre-expectation} $\symCwp$ on \emph{conditional expectations} which are pairs of expectations $(\expA, \expB) \in \CExpectations$.
For each state $\wpState$, the conditional expectation $(\expA, \expB)$ can be evaluated as $\frac{\expA(\wpState)}{\expB(\wpState)}$ provided $\expB(\wpState) \neq 0$.
We use the notation $\condforce{(\expA, \expB)} = \frac{\expA}{\expB}$ to denote the evaluation of a conditional expectation $(\expA, \expB)$ into an expectation when it is defined.
For a statement $\stmt$ and a conditional expectation $(\expA \times \expB) \in \CExpectations$, we define the function $\symCwp$ simply as the pair of two $\symWp$ computations:
\begin{equation}
    \cwp{\stmt} \colon \CExpectations \to \CExpectations,~ \cwp{\stmt}(\expA, \expB) = (\wp{\stmt}(X), \wp{\stmt}(\expB))~.
\end{equation}
For a given initial state $\wpState \in \wpStates$, the expression $\cwp{\stmt}(\expA, \expB)(\wpState)$ is the conditional expectation on termination of $S$ of $\expA$ given $\expB$.
Let $\expOne$ denote the expectation that maps all states to the value $1$.
Then, $\cwp{\stmt}(\expA, \expOne)(\wpState)$ represents the expected value of $\expA$ after running $\stmt$ when starting in state $\wpState$, conditioned on all $\symObserve$ statements succeeding.
This is because $\wp{\stmt}(\expOne)(\wpState)$ evaluates to the probability of all observations in the computation starting from $\wpState$ succeeding.

The definition of $\symCwp$ above is a simplified version of the one given by Jansen and coauthors~\cite{DBLP:journals/entcs/0001KKOGM15}.
We use the fact that all programs in our language terminate certainly and therefore do not need to condition on termination.
In the original work, this is done by dividing by weakest \emph{liberal} pre-expectations to obtain the probability of $\symObserve$ failure or nontermination.

\subsection{Continuous Symbolic Semantics: SOGA}
\label{sec:continouous_semantics}

For a random variable $\rvx$ and a distribution $D$, let $\density{\rvx}$ and $\density{D}$ denote their density functions. A normal random variable with mean $\mu$ and covariance matrix $\Sigma$ is denoted by $\normal(\mu, \Sigma)$ and has density function $\density{\normal(\mu, \Sigma)}.$ Let $x[x_i \to c]$ denote the vector with $c$ as $i$-th component, and every other component equal to $x_j$. Let $\marg{y}{D}$ denote the marginal of distribution $D$ with respect to subvector $y$. For a predicate $\bexp$,  $\prob{D}{\bexp}$ denotes the probability of $\bexp$ under distribution $D$ and $\cond{D}{\bexp}$ denotes distribution $D$ conditioned to $\bexp$.

We assume that the program is defined over program variables \(x_1, \dots, x_n\), which we arrange in a vector $x = (x_1, \hdots, x_n)$ taking values in $\mathbb{R}^n$. No other variables are defined in the program, so all variables are statically detected and included in the vector. The set of program states is given by the set of (possibly degenerate) Gaussian mixtures over $\mathbb{R}^n$ denoted as
\begin{equation}
  \gms ::= G \sim \sum_{i=1}^C \pi_i\, \normal(\mu_i, \sigma_i )~,  \enspace
\end{equation}
where $C \in \mathbb{N}_{\ge 0}$ denotes the number of \emph{components} in the mixture and $\pi_i$ for $i \in \{1, \hdots, C\}$ are real numbers in $[0,1]$ such that $\sum_{i=1}^C \pi_i =1$ called \emph{weights}. For simplicity, we use a multivariate Dirac delta centered at the zero vector to be the initial state for all programs, i.e.\ a degenerate Gaussian distribution with mean $0$ and the null covariance matrix.

Expressions can either be the product between two variables, or linear expressions on program variables, i.e.
\begin{equation}
  \mathbb{E} = \{ x_ix_j \mid i, j = 1 \dots n \} \cup \{ c_1x_1 + \hdots + c_nx_n + c_0 \mid c_i \in \mathbb{R} \} \enspace .
\end{equation}
Finally, for continuous programs, we restrict the set of distributions, which can be assigned to variables to the set of distributions with constant parameters, i.e. $ \dists (\overline\theta) = \dists $.

We define the SOGA semantics of a program $\stmt$ denotationally, as a forward operator $\soga{\stmt} \colon \gms \to \gms$. The operator is defined by structural induction on program structure, and uses two others:%
\begin{itemize}

    \item $\exact{\stmt} \colon \gms \to \gms$, which computes the exact denotational semantics of a probabilistic program \cite{kozen1979semantics}, itself defined recursively on the program syntax;

    \item $\mmop \colon \dists \to \gms$ is the moment-matching operator. Intuitively, \(\mmop\) returns a Gaussian mixture approximation of a density function, matching its first two moments. When $\mmop$ acts on a mixture, it acts on every component; when it acts on a single-component mixture, it maps it to a Gaussian having mean and covariance matrix matching the original distribution. Formally:
    \begin{equation}
      \mmop(D) = \begin{cases} \pi_1 \mmop(D_1) + \hdots + \pi_C \mmop(D_C) & \text{if } D=\sum_{i=1}^C \pi_iD_i \\ \normal(\mu_D, \Sigma_D) & \text{else.} \end{cases} \enspace .
    \end{equation}
    Notably, when acting on Gaussian mixtures, $\mmop$ leaves them unaltered\,\cite{RandoneBIT24}.
    \looseness -1

\end{itemize}
The recursive definition of $\soga{\stmt}$ and $\exact{\stmt}$ are given in \Cref{tab:gaussian-semantics}.

\begin{table}[t]
  \renewcommand \tabcolsep {4pt}
  \begin{tabularx}{\linewidth}{
    @{}
    >{\footnotesize}l
    >{\footnotesize}l
    >{\footnotesize}X
    @{}
  }
    \toprule
    $\stmt$ & $\exact{\stmt}(D)$ & $\soga{\stmt}(\gm)$ \\
    \midrule
    $\stmtSkip$ & $D$ & $\gm$ \\
    $\stmtAsgn{x_i}{\aexp}$ & $D' \text{ s.t. } x[x_i \to \aexp] \sim D'$ & $\mmop(\exact{\stmt}(\gm))$\\

    $\stmtDist{x_i}{D'}$ & $\marg{x \setminus x_i}{D} \otimes D'$ & $\marg{x \setminus x_i}{\gm} \otimes \mmop(D')$ \\
    $\symObserve(\bexp)$ & $\cond{D}{\bexp}$ & $\mmop(\cond{\gm}{\bexp})$ \\
    $\stmtSeq{\stmt_1}{\stmt_2}$ & $\exact{\stmt_2}(\exact{\stmt_1}(D))$ & $\soga{\stmt_2}(\soga{\stmt_1}(\gm))$ \\
    $\stmtFor{i}{n}{\stmt}$ & $\exact{\stmt}^n(D)$ & $\soga{\stmt}^n(\gm)$ \\

      $\stmtIf{\bexp}{\!\stmt_1\!}{\!\stmt_2\!}$
    & $\prob{D}{\bexp}\exact{\stmt_1}(\cond{D}{\bexp}) $
    & $\prob{\gm}{\bexp}\soga{\stmt_1}(\mmop(\cond{\gm}{\bexp})) $ \\[-1.5mm]

    & $\strut~ + \prob{D}{\neg\bexp}  \exact{\stmt_2}(\cond{D}{\neg \bexp})$
    & $\strut~ + \prob{\gm}{\neg\bexp} \soga{\stmt_2}(\mmop(\cond{\gm}{\neg \bexp}))$ \\
    \bottomrule
\end{tabularx}

\bigskip

\caption{Rules for the recursive definition of $\exact{\stmt}$ and $\soga{\stmt}$}
\label{tab:gaussian-semantics}
\end{table}


%% file: metrics.tex
\section{Symbolic Computation of Leakage Metrics}%
\label{sec:metrics}

\subsection{Basic definitions}

\paragraph {Entropy.}

Shannon entropy measures uncertainty of a random variable. The entropy of a random variable $\rvx$, denoted as $\entropy(\rvx)$,  is the negative expected $\log$ probability of each of its outcomes:
\begin{equation}
  \entropy(\rvx) = - \sum_x p(x)\log(p(x)) \enspace ,
\end{equation}
where $x \in \supportx$ is the support of $\rvx$ and $p(x)$ denotes its probability mass function~(pmf).
The continuous version, known as differential entropy, simply replaces the summation with an integral:
\begin{equation}
\entropy(\rvx) = - \int_x p(x)\log(p(x))dx \enspace .
\end{equation}
The differential entropy does not always behave well; for instance it is not guaranteed to be positive.  We ignore these issues in the present work.

Given a joint distribution $p(x,y)$ of two random variables $\rvx$ and $\rvy$, entropy can also be computed for a distribution, when one marginal is conditioned to some predicate $\varphi(\rvy)$, i.e., $\entropy(X \mid \varphi(\rvy))$. Examples of posterior distributions include $\rvx \mid \rvy = c$ and $\rvx \mid \rvy < c$. This can be used to model learning about secret variables from observations of public variables made by an attacker.

In Privug, a program \( \program : \inputs \to \outputs \) induces a joint probability distribution $p(\rvi, \rvo)$---if we lift it to the probability monad, so to run on a random, instead of concrete, variable $I$. In particular, we are interested in the uncertainty of the posterior distribution over the input, after observing (conditioning on) the output. For a program that returns a real number $c$, we use $\entropy(\rvi \mid \rvo = c)$ to quantify the uncertainty of the attacker about the input after observing $c$.

\paragraph{Conditional entropy,} denoted $\entropy(\rvx \mid \rvy)$, quantifies the amount of information needed to describe a random variable $\rvx$ given that the value of $\rvy$ is known. Formally, conditional entropy is defined as
\begin{equation}
  \entropy(\rvx \mid \rvy) = - \sum_{x,y} p(x,y)\log \frac{p(x,y)}{p(y)} \enspace .
\end{equation}
Note that this is different from the entropy of a posterior distribution $\entropy(\rvx \mid \varphi(\rvy))$. Conditional entropy summarizes the posterior entropy over all possible observations.  It is a useful metric when we are not interested in analyzing the risks of releasing a particular output, but analyzing risks for the program as a whole. By computing $\entropy(\rvi \mid \rvo)$, we quantify the attacker's expected uncertainty about (secret) inputs given that she observes any possible output.

\paragraph{KL divergence.}
Consider two random variables $P$ and $Q$ with pmfs $p(x)$ and $q(x)$ over the same sample space $\Omega_P = \Omega_Q$.  The KL divergence of $P$ from $Q$ is the amount of information needed to reconstruct $P$ if $Q$ is known. It is a common distance measure between distributions. Formally, it is defined as
\begin{equation}
   \kl(P \, || \, Q) = \sum_x p(x) \log \frac{p(x)}{q(x)} \enspace .
\end{equation}
KL divergence is not formally a distance metric, as it is not symmetric. As before, for continuous random variables we simply replace the summation with an integral:
\begin{equation}
  \kl(P || Q) = \int_x p(x) \log \frac{p(x)}{q(x)} dx \enspace .
\end{equation}
In Bayesian inference, the KL divergence of a posterior $P$ and a prior $Q$ expresses the amount of \emph{information gained} by revising one's beliefs from the prior distribution $Q$ to the posterior distribution $P$.
In the setting of quantitative information flow for a analyzing a program with input $\rvi$ and output $\rvo$, we compute $\kl(\rvi \mid \rvo = c \, || \, \rvi)$ to quantify how much an attacker learns from observing output $c$.

\newcommand{\mi}{\textrm{I}}
\paragraph{Mutual information.}
Given two random variables, $\rvx$ and $\rvy$, mutual information, written $\mi(\rvx; \rvy)$, is a measure characterizing how much information two random variables share; in this sense it is also a measure of independence between random variables. If $\mi(\rvx, \rvy) = 0$, then $\rvx$ and $\rvy$ are independent random variables. Mutual information is defined as
\begin{equation}
  \mi(\rvx; \rvy) = \entropy(\rvx) - \entropy(\rvx \mid \rvy) \enspace .
\end{equation}
We use mutual information to assess how independent is the output of a program $\rvo$ from the (secret) inputs $\rvi$. To this end, we compute $\mi(\rvi; \rvo)$ and check how far it is from zero. In general, it is difficult to interpret the magnitude of $\mi(\rvi; \rvo)$ as it corresponds to the intersection of the entropy of $\rvi$ and $\rvo$, i.e., $\mi(\rvi; \rvo) < \max(\entropy(\rvi), \entropy(\rvo))$.
As a consequence, a low value of $\mi(\rvi; \rvo)$ could be associated with inherent low entropy of $\rvi$ or $\rvo$.

\subsection{Discrete WPE Semantics}

\paragraph{Entropy.}
In the weakest pre-expectation framework, the entropy can be computed as a composition of weakest pre-expectation transformers.
We focus on the entropy of the distribution of values of variables on program termination.
This entropy is conditioned on all $\symObserve$ statements succeeding.

\begin{proposition}\label{proposition:wp-entropy}
  For a program $\stmt$ and variable $x_i$, the entropy $\entropy(x_i)$ of the value of $x_i$ on termination of $\stmt$ can be computed as follows:
  \begin{align*}
    \entropy(x_i) \quad&=\quad \cwp{\stmt}(\alog\condforce{\cwp{\stmt}(\iverson{\aux{x_i} = x_i}, \expOne)\substBy{\aux{x_i}}{x_i}}, \expOne)~.
  \end{align*}
\end{proposition}
In the above, we use the inner $\symCwp$ transformer to compute the uncertainty for specific values of the variable $x_i$. The result of the inner $\symCwp$ computation is an expectation with a free \emph{ghost variable} $\aux{x_i}$ whose value is checked to be equal to the value of $x_i$ at the end of the program. We rename all occurrences of $\aux{x_i}$ to $x_i$. Then, we use the outer $\symCwp$ transformer again to compute the expected uncertainty, so that the expected uncertainty is computed based on the probabilities for $x_i$ in the unconditioned distribution of final states of the program, i.e. where $\symObserve$ statements are ignored. Note that we use the convention that $\alog(0) = \infty$ to ensure that intermediate expectations are well-defined.

\begin{proof}
  Let $\probCondTerm{\stmt}{\wpState}{x_i}{v_i}$ denote the probability of $\stmt$ terminating in a state $\wpState$ with $\wpState(x_i) = v_i$, conditioned on $\symObserve$ statements succeeding.
  We can express this probability using conditional weakest pre-expectations as the expected value on termination of the Iverson bracket $[x_i = v_i]$, i.e.
  \begin{align}
    \probCondTerm{\stmt}{\wpState}{x_i}{v_i} \quad&=\quad \condforce{\cwp{\stmt}([x_i = v_i], \expOne)}(\wpState)~. \label{eqn:wp-probdist}
  \end{align}
  Therefore, $\cwp{\stmt}([\aux{x_i} = x_i], \expOne)(\wpState)$ is the conditional expectation that evaluates to the conditional probability of $\stmt$ terminating in a state where $\aux{x_i} = x_i$.

  Let $\expB = \cwp{\stmt}([\aux{x_i} = x_i], \expOne)\substBy{\aux{x_i}}{x_i}$.
  By performing the substitution $\aux{x_i} \mapsto {x_i}$ in $\cwp{\stmt}([\aux{x_i} = x_i], \expOne)$, we obtain an expectation where $\expB(\wpState)$ represents the conditional probability of starting in a state $\wpState$ and terminating in a state $\wpState'$ in which $\wpState'(x_i)$ equals the initial value $\wpState(x_i)$ of $x_i$.

  To obtain \Cref{proposition:wp-entropy}, we compute the expected value of the negated logarithm of $\expB$ after running the program.
  The negation ensures that $\alog(\expB)$ is nonnegative and thus a valid expectation.
  \begin{align*}
    &\phantom{=}\quad \cwp{\stmt}(\alog\condforce{\cwp{\stmt}(\iverson{\aux{x_i} = x_i}, \expOne)}\substBy{\aux{x_i}}{x_i}, \expOne) \tag{\Cref{proposition:wp-entropy}} \\
    &=\quad \sum_{v_i} \probCondTerm{\stmt}{\wpState}{x_i}{v_i} \cdot (\alog\condforce{\cwp{\stmt}(\iverson{\aux{x_i} = x_i}, \expOne)}\substBy{\aux{x_i}}{v_i}) \\
    &=\quad \sum_{v_i} \probCondTerm{\stmt}{\wpState}{x_i}{v_i} \cdot (\alog(\probCondTerm{\stmt}{\wpState}{x_i}{v_i})) \tag{\Cref{eqn:wp-probdist}} \\
    &=\quad -\sum_{v_i} \probCondTerm{\stmt}{\wpState}{x_i}{v_i} \cdot \palog(\probCondTerm{\stmt}{\wpState}{x_i}{v_i}) \quad= \entropy(x_i)~. \tag*{\qed}
  \end{align*}
\end{proof}

\paragraph{Conditional entropy.}
Whereas \Cref{proposition:wp-entropy} uses weakest pre-expectations to compute the entropy conditioned on all $\symObserve$ statements succeeding, \Cref{proposition:wp-cond-entropy} below obtains the conditional entropy for a variable $y_i$ given that the value of $x_i$ is known \emph{and} that all $\symObserve$ statements succeed.

\begin{proposition}\label{proposition:wp-cond-entropy}
  For a program $\stmt$ and variables $x_i$ and $y_i$, the conditional entropy $\entropy(y_i \mid x_i)$ of the value of $y_i$ on termination of $\stmt$ given $x_i$ can be computed as follows:
  \begin{align*}
    \entropy(y_i \mid x_i) = \cwp{\stmt}(\alog \condforceStart
      \cwp{\stmt}(&
       \quad \iverson{\aux{x_i} = x_i \land \aux{y_i} = y_i}, \\
      &\quad \iverson{\aux{x_i} = x_i} \\
      &)\condforceEnd\substBy{\aux{x_i}}{x_i}\substBy{\aux{y_i}}{y_i}, \expOne
      )~.
  \end{align*}
\end{proposition}

\noindent
The proof is analogous to the above one. The only modification is in the inner conditional weakest pre-expectation computation, where we compute with respect to the probability of terminating with the value of $y_i$ given the value for $x_i$.
\looseness -1

\paragraph{KL divergence.}

To compute the KL divergence in the weakest pre-expectation framework, we can again use a similar construction to the above ones.
It is the conditional expected value over $x_i$ (represented as $\aux{x_i}$) of the expected logarithmic difference between the probabilities of events $\iverson{\aux{x_i} = x_i}$ and $\iverson{\aux{x_i} = y_i}$.
The logarithmic difference is given by the logarithm of the nested $\symCwp$ computation.

\begin{proposition}\label{proposition:wp-kl-divergence}
  Let $\stmt$ be a program and let $x_i$ and $y_i$ be program variables such that $y_i$ is not an input variable in $\stmt$.
  The KL divergence $\kl(x_i \doublemid y_i)$ between values of $x_i$ and $y_i$ on termination and conditioned on all observations succeeding is given by:
  \begin{align*}
    \kl(x_i \doublemid y_i) &= \cwp{\stmt}\left(\palog\left\condforceStart
      \cwp{\stmt}(\iverson{\aux{x_i} = x_i}, \iverson{\aux{x_i} = y_i})\substBy{\aux{x_i}}{x_i}, \expOne
      \right)\right\condforceEnd~.
  \end{align*}
\end{proposition}
The proof is analogous to the ones of the previous propositions.

\subsection{Continuous GMM semantics}
\label{subsec:continuous_gmm_semantics}

The situation for the continuous programs is somewhat more complex.  Instead of computing exact values of the metrics, we will approximate them.  We begin with a negative result, stating that, in general, leakage metrics computed on the approximated output given by the SOGA semantics are not comparable with leakage metrics computed on the exact output.

\begin{proposition} \label{prop:negative}
    $\entropy(\exact{\stmt}(\gm))$ is not comparable to $\entropy(\soga{\stmt}(\gm))$.
\end{proposition}

\begin{proof}
    We use three programs to show that $\entropy(\exact{\stmt}(\gm))$ can be smaller, equal or greater than $\entropy(\soga{\stmt}(\gm))$. In all three examples we assume as output variable $x_3$.
    Consider the programs:
    \begin{align*}
    \stmt_1 &:: = x_1 \sim \normal(0,1); x_2 \sim \normal(0,1); x_3 := x_1x_2 \\
    \stmt_2 &:: = x_1 \sim \normal(0,1); x_2 \sim \normal(0,1); x_3 := x_1+x_2 \\
    \stmt_3 &:: = x_1 \sim 0.25\normal(0,1)+0.75\normal(0,2); x_2 \sim \normal(0,1); x_3 := x_1x_2
    \end{align*}

    \noindent
    In program $\stmt_1$, the marginal over $x_3$ computed by $\exact{\stmt_1}$ is the product of two independent standard Gaussian distributions with density $p(x_3) = \frac{1}{\pi}K_0(|x_3|)$, where $K_0$ is a modified Bessel function of the second kind \cite{weisstein}. The marginal computed by $\soga{\stmt_1}$ instead is a Gaussian having same mean and variance as the exact marginal. It is possible to verify by integration that in this case $\entropy(\exact{\stmt_1}(G)) < \entropy(\soga{\stmt_1}(G))$ for any $G$.

    For program $S_2$, $\exact{\stmt_2}(G) = \soga{\stmt_2}(G)$ for any $G$. Consequently, the equality holds also for the entropies.

    Finally, for program $\stmt_3$, the marginal over $x_3$ computed by $\exact{\stmt_3}$ is the product of a Gaussian mixture and a standard Gaussian distributions, expressible as a mixture of products of Gaussians. The marginal computed by $\soga{\stmt_3}$ instead is a Gaussian mixture in which each component moment-matches one component of the exact marginal. Computing the integrals it is possible to verify by integration that $\entropy(\exact{\stmt_3}(G)) > \entropy(\soga{\stmt_3}(G))$ for any $G$.
    \qed
\end{proof}

Since in general computing leakage metrics on SOGA outputs is not informative, we will restrict our attention to those programs for which the semantics computed by SOGA coincides with the exact semantics. In order for this to be true, we have to rule out all program instructions that introduce approximations, such as products and truncations of continuous random variables. The following proposition gives sufficient conditions for the exactness of SOGA semantics.

\begin{proposition} \label{prop:exact}
    Let \begin {align*}
  \stmt ::= ~ & \stmtSkip \mid \stmtAsgn{x}{\aexp} \mid \stmtDist{x}{\gm'} \mid \symObserve(\bexp_{\textit{obs}}) \mid \stmtSeq{\stmt}{\stmt} \mid \stmtIf{\bexp_{\textit{if}}}{\stmt}{\stmt}
\end {align*}
where:
\begin{itemize}
    \item[i)] $\aexp$ is linear;
    \item[ii)] $\gm' \in \gms$;
    \item[iii)] $\bexp_{\textit{obs}}$ is either an equality or only depends on discrete random variables;
    \item[iv)] $\bexp_{\textit{if}}$ only depends on discrete random variables.
\end{itemize}

\noindent
Then, for any $\gm \in \gms$ we have \(\exact{\stmt}(\gm) = \soga{\stmt}(\gm).\)
\end{proposition}

\begin{proof}
By structural induction on $\stmt$.
\begin{itemize}

    \item $\stmt ::= \stmtSkip$. Conclusion follows from semantics definitions.

    \item $\stmt ::= \stmtAsgn{x}{\aexp}$, $\aexp$ linear. Since $\gm$ is a GM, $\exact{\stmt}(\gm)$ is also GM. Since $\mmop$ leaves GMs unaltered we have $\soga{\stmt}(\gm) = \mmop(\exact{\stmt}(\gm)) = \exact{\stmt}(\gm)$.

    \item $\stmt :: = \stmtDist{x}{\gm'}$, $\gm' \in \gms$. Since $\gm'$ is a GM, $\mmop$ leaves that unaltered and conclusion follows from semantics definitions.

    \item $\stmt :: = \symObserve(\bexp)$, where $\bexp$ is either an equality or depends only on discrete random variables. In both cases $\cond{\gm}{\bexp}$ is still a GM: in fact, a Gaussian mixture where one of the coordinates is conditioned to have a fixed value it is still a Gaussian mixture, though degenerate \cite{bishop2006pattern}; on the other hand, if $\bexp$ only depends on discrete random variables $\cond{\gm}{\bexp}$ is still a Gaussian mixture obtained from $\gm$ by selecting only the components where $\bexp$ is satisfied. Since GMs are left unaltered by $\mmop$, the conclusion follows from the respective semantics.

    \item $\stmt :: = \stmtSeq{\stmt_1}{\stmt_2}$ and suppose the statement holds for $\stmt_1, \stmt_2$. Then, by repeated applications of the inductive hypothesis:
    \begin{align*}
    \exact{\stmt}(\gm) = \exact{\stmt_2}(\exact{\stmt_1}(\gm)) & = \exact{\stmt_2}(\soga{\stmt_1}(\gm)) \\
    & = \soga{\stmt_2}(\soga{\stmt_1}(\gm)).
    \end{align*}

    \item $\stmtIf{\bexp}{\stmt}{\stmt}$, $\bexp$ only dependent on discrete random variables and suppose the statement holds for $\stmt_1, \stmt_2$. Then:
    \begin{align*}
    \exact{\stmt}(\gm) & = \prob{\gm}{\bexp}\cdot\exact{\stmt_1}(\cond{\gm}{\bexp}) + \prob{\gm}{\neg\bexp}\cdot\exact{\stmt_2}(\cond{\gm}{\neg \bexp}) \\
    & = \prob{\gm}{\bexp}\cdot\soga{\stmt_1}(\cond{\gm}{\bexp}) + \prob{\gm}{\neg\bexp}\cdot\soga{\stmt_2}(\cond{\gm}{\neg \bexp}) \\
    & = \soga{\stmt}(G)
    \end{align*}
    where the inductive hypothesis holds for $\exact{\stmt_1}(G)$ and $\exact{\stmt_2}(G)$ due to the fact that, since $\bexp$ only depends on discrete random variables, $\cond{\gm}{\bexp}$ and $\cond{\gm}{\neg \bexp}$ are still GMs.\qed
\end{itemize}
\end{proof}

\noindent
For programs in the scope of \Cref{prop:exact} we show how upper and lower bound on different leakage metrics can be computed.

\paragraph{Entropy.}

For a Gaussian random vector $\rvx \sim \normal(\mu, \Sigma)$ the entropy can be computed analytically as $\entropy(\rvx) = \frac{1}{2} \log((2\pi e)^n | \Sigma|)$. For Gaussian mixtures no exact analytical form is available for the entropy, however several lower and upper bounds can be derived.

The simplest upper bound, is given by the entropy of a Gaussian matching the mean and the covariance matrix of the mixture, exploiting the maximum entropy property of Gaussian distributions. (A Gaussian \(\normal(\mu,\sigma)\) has the highest entropy among all densities with mean \(\mu\) and standard deviation \(\sigma\).) The following proposition puts together results from \cite{huber2008entropy} to provide a lower and a tighter upper bound, which can be further refined exploiting the Gaussian approximation on selected groups of components.

\begin{proposition}[Bounds for entropy of GMs \cite{huber2008entropy}] \label{prop:Hbounds}

\begin{itemize}
    \item[a)] Let $\rvx = \sum_{i=1}^C \pi_i \normal(\mu_i, \Sigma_i)$ and $z_{i,j} = \phi(\mu_i; \mu_j, \Sigma_i + \Sigma_j)$, where $\phi(x; \mu, \Sigma)$ denotes the density of a Gaussian with mean $\mu$ and covariance matrix $\Sigma$. Then
    $$\entropy(\rvx) \ge -\sum_{i=1}^C \pi_i \log \left( \sum_{j=1}^C \pi_j z_{i,j} \right).$$

    \item[b)] Let $\rvx = \sum_{i=1}^C \pi_i \normal(\mu_i, \Sigma_i)$. Then
    $$\entropy(\rvx) \le -\sum_{i=1}^C \pi_i \log(\pi_i) + \sum_{i=1}^C \pi_i \entropy(\normal(\mu_i, \Sigma_i));$$
   Moreover, if $\rvx_1 = \sum_{i=1}^k \pi_i \normal(\mu_i, \Sigma_i)$ and $\rvx_2 = \sum_{i=k+1}^C \pi_i \normal(\mu_i, \Sigma_i)$ for some $k = 1, \hdots, C$, let  $\normal(\bar{\mu}, \bar{\Sigma})$ be a Gaussian with same mean and covariance matrix as $\rvx_1$ and $\pi = \sum_{i=1}^k \pi_i$. Then
    $$\entropy(\rvx) \le \entropy(\pi\normal(\bar{\mu},\bar{\Sigma}) + \rvx_2).$$
\end{itemize}
\end{proposition}

\noindent
Upper and lower bounds on the conditional differential entropy can be derived using the relation  $\entropy(\rvx \mid \rvy) = \entropy(\rvx, \rvy) - \entropy(\rvy)$. In particular, one can consider a lower bound $\entropy(\rvx, \rvy)$ and an upper bound on $\entropy(\rvy)$. The difference between the two will yield a lower bound on $\entropy(\rvx \mid \rvy)$. Similarly, considering the difference between an upper bound for $\entropy(\rvx, \rvy)$ and a lower bound for $\entropy(\rvy)$ will yield an upper bound on $\entropy(\rvx \mid \rvy)$.


\paragraph{KL divergence.}

If $\rvx \sim \normal(\mu_X, \Sigma_X)$ and $\rvy \sim \normal(\mu_Y, \Sigma_Y)$ are Gaussian random vectors on $\mathbb{R}^n$, their KL divergence can be computed in closed form as:
$$\kl( \rvx || \rvy) = \frac{1}{2}\left[ \log \frac{|\Sigma_Y|}{|\Sigma_X|} + \text{Tr}(\Sigma_Y^{-1}\Sigma_X) - n + (\mu_X - \mu_Y)^T \Sigma_Y^{-1}(\mu_X - \mu_Y)\right].$$

If $\rvx$ and $\rvy$ are Gaussian mixtures, no analytical form is available; however upper and lower bounds can be derived analytically. We start by observing that we can express the KL divergence between $\rvx$ and $\rvy$ as
\begin{align*} \kl( \rvx || \rvy) = \int f_{\rvx}(x) \log\left( \frac{ f_{\rvx}(x) }{ f_{\rvy}(x) } \right) dx =  - \entropy(\rvx) - \underbrace{\int f_{\rvx}(x) \log\left( f_{\rvy}(x) \right) dx}_{:=L_{\rvx}(\rvy)}~,
\end{align*}
where $f_{\rvx}$ and $f_{\rvy}$ are the density of $\rvx$ and $\rvy$ respectively. Observe that in the case of Gaussians, because we can compute $\kl( \rvx || \rvy)$ and $\entropy(\rvx)$ analytically, we can also compute $L_\rvx(\rvy)$. In the case of Gaussian mixtures, since we have already provided upper and lower bounds for $\entropy(\rvx)$, then if we are able to provide upper and lower bounds for $L_{\rvx}(\rvy)$, we can easily bound the KL divergence as well. The following proposition, adapted from \cite{hershey2007approximating}, provides bounds for this term.

\begin{proposition}[Bounds for KL divergence \cite{hershey2007approximating}] \label{prop:Lbounds}
Let $\rvx \sim \sum_{a=1}^{A} \pi_a \rvx_a$, $\rvy \sim \sum_{b=1}^{B} \rho_b \rvy_b$ with $\rvx_a \sim \normal(\mu_a, \Sigma_a)$ and $\rvy_b \sim \normal(\mu_b, \Sigma_b)$. Then
\begin{itemize}
    \item For $z_{a,b} = \phi(\mu_a; \mu_b, \Sigma_a+\Sigma_b)$ it holds
    $$L_{\rvx}(\rvy) \le \sum_{a=1}^A \pi_a \log \left( \sum_{b=1}^B \rho_b z_{a,b} \right).$$
    \item For any $\phi_{a,b} \ge 0$, $\sum_{b=1}^B \phi_{a,b} = 1$ it holds
    $$ L_{\rvx}(\rvy) \ge \sum_{a=1}^A \sum_{b=1}^B \pi_a \phi_{b,a} \left[ \log\left( \frac{\rho_b}{\phi_{b,a}}\right) + L_{\rvx_a}(\rvy_b) \right]$$
    and the lower bound is maximized when $\phi_{b,a} = \frac{\rho_b e^{-\kl(\rvx_a || \rvy_b)}}{\sum_{b'=1}^B \rho_{b'} e^{-\kl(\rvx_a || \rvy_b)}}$.
\end{itemize}
\end{proposition}
\begin{proof}
    Let us first prove the upper bound:
    \begin{align*} L_{\rvx}(\rvy) & = \int \density{\rvx}(x) \log(\density{\rvy}(x)) dx \\
    & = \sum_{a=1}^A \pi_a \int \normal( \mu_a, \Sigma_a) \log \left( \sum_{b=1}^B \rho_b \normal( \mu_b, \Sigma_b ) \right) dx \tag{Jensen's ineq.} \\
    & \le \sum_{a=1}^A \pi_a \log \left( \sum_{b=1}^B \rho_b \underbrace{\int \normal( \mu_a, \Sigma_a)  \normal( \mu_b, \Sigma_b )  dx}_{= z_{a,b}} \right).
    \end{align*}
    For the lower bound:
    \begin{align*} L_{\rvx}(\rvy) & = \int \density{\rvx}(x) \log(\density{\rvy}(x)) dx = \\
    & = \sum_{a=1}^A \pi_a \int \normal( \mu_a, \Sigma_a) \log \left( \sum_{b=1}^B \rho_b \normal( \mu_b, \Sigma_b ) \right) dx \\
    & = \sum_{a=1}^A \pi_a \int \normal( \mu_a, \Sigma_a) \log \left( \sum_{b=1}^B \frac{\rho_b\phi_{a,b}}{\phi_{a,b}} \normal( \mu_b, \Sigma_b ) \right) dx \tag{concavity}  \\
    & \ge \sum_{a=1}^A \pi_a  \sum_{b=1}^B
    \phi_{b,a} \int \normal( \mu_a, \Sigma_a) \log \left( \frac{\rho_b}{\phi_{a,b}} \normal( \mu_b, \Sigma_b ) \right) dx \\
    & = \sum_{a=1}^A \sum_{b=1}^B
    \pi_a \phi_{b,a} \left[ \log \left( \frac{\rho_b}{\phi_{a,b}}\right) + \int \normal(\mu_a, \Sigma_a) \log\left( \normal( \mu_b, \Sigma_b ) \right) dx \right].
    \end{align*}
\end{proof}

\noindent
In \cite{hershey2007approximating}, a direct upper bound on $\kl(\rvx || \rvy)$ is derived variationally.

\begin{proposition}[Variational upper bound for KL \cite{hershey2007approximating}] \\
Let $\rvx \sim \sum_{a=1}^{A} \pi_a \rvx_a, \rvy \sim \sum_{b=1}^{B} \rho_b \rvy_b$ with $\rvx_a \sim \normal(\mu_a, \Sigma_a)$ and $\rvy_b \sim \normal(\mu_b, \Sigma_b)$. Then, for any $\phi_{b|a}, \psi_{a|b} \ge 0$ such that $\sum_{b=1}^B \phi_{b|a} = \pi_a, \sum_{a=1}^A \psi_{a|b} = \rho_b$ it holds:
$$\kl(\rvx || \rvy) \le \sum_{a=1}^A \sum_{b=1}^B  \phi_{b|a} \left[ \log\left(\frac{\phi_{b|a}}{\psi_{a|b}}\right) + \kl( \rvx_a || \rvy_b) \right]~. $$

Moreover, the upper bound is minimized by $\hat{\phi}_{b|a}, \hat{\psi}_{a|b}$ satisfying
$$  \hat{\psi}_{a|b} = \frac{\rho_b \hat{\phi}_{b|a} }{\sum_{a'=1}^A \hat{\phi}_{b|a'}}, \quad \hat{\phi}_{b|a} = \frac{\pi_a  \hat{\psi}_{a|b} e^{-\kl( \rvx_a || \rvy_b)}}{\sum_{b'=1}^B \hat{\psi}_{a|b'} e^{-\kl( \rvx_a || \rvy_{b'})}}~. $$
\end{proposition}
Observe that from the previous proposition, setting $\phi_{b|a} = \psi_{a|b} = \pi_a\rho_b$ one obtains the upper bound $\kl(\rvx || \rvy) \le \sum_{a=1}^A \sum_{b=1}^B \pi_a\rho_b \kl(\rvx_a || \rvy_b)~.$

\paragraph{Mutual information.}

Bounds on mutual information can be recovered from the previous results. In particular, one can use the bounds on entropy provided by \Cref{prop:Hbounds} in the relation
$$ \mi(\rvx ; \rvy) = \entropy(\rvx) - \entropy(\rvx | \rvy)~.$$
Alternatively one can exploit the bounds given by \Cref{prop:Lbounds} in the relation:
$$ \mi(\rvx ; \rvy) = \entropy(\rvx, \rvy) - L_{(\rvx, \rvy)}(\rvx) - L_{(\rvx, \rvy)}(\rvy).(\rvy)~.$$


%% file: case_studies.tex
\section{Case Studies}\label{sec:case_studies}

\subsection{Binary Randomized Response (discrete)}
\label{subsec:randomized_response}

\newcommand{\brrResponseProb}{p}
\newcommand{\brrResponse}{r}
\newcommand{\brrRandomizeProb}{\frac{e^\epsilon}{e^\epsilon + 1}}
\newcommand{\brrRandomize}{t}
\newcommand{\brrResult}{x}
\newcommand{\brrOutput}{o}

We consider a randomized response mechanism that satisfies $\epsilon$-differential privacy.
Let $\vec{\response}$ denote a binary vector with responses $\response_i \in \{0,1\}$ from individuals $1,\ldots,n$.
The randomized response mechanism, denoted as $M(\vec{r})$, aims to protect each individual response $\response_i$ (up to some privacy level $\epsilon$, see below for details), given that the responses are used to compute an output $o = f(M(\vec{\response}))$ that will be accessible to the attacker.
Let $M(\vec{\response})$ randomize each response $M(\response_i)$ as follows:
$$
p(M(\response_i) = v) =
\begin{cases}
\frac{e^\epsilon}{e^\epsilon + 1} & \text{if } v  = \response_i,\\
1 - \frac{e^\epsilon}{e^\epsilon + 1} & \text{otherwise~.}
\end{cases}
$$
It is well-known that this mechanism ensures $\epsilon$-differential privacy:
Consider two neighboring datasets $\vec{\response}, \vec{\response}'$ that differ in record $\response_v \not = \response'_v$, and output $\vec{a}$, then the following holds:
$$
\frac{p(M(\vec{\response}) = \vec{a})}{p(M(\vec{\response}') = \vec{a})} =
\frac{\prod_i p(M(\response_i) = a_i)}{\prod_i p(M(\response'_i) = a_i)} \cdot \frac{p(M(\response_v) = a_v)}{p(M(\response'_v) = a_v)}~.
$$
In the case that $a_v  = \response_v$, we have
$$
\frac{e^\epsilon / e^\epsilon + 1}{1 - (e^\epsilon / e^\epsilon + 1)} =
e^\epsilon~.
$$
Otherwise, if $a_v  \not= \response_v$, then
$$
\frac{1 - (e^\epsilon / e^\epsilon + 1)}{e^\epsilon / e^\epsilon + 1} = 1/e^\epsilon~.
$$
Thus, $p(M(\vec{\response}) = \vec{a}) / p(M(\vec{\response}') = \vec{a}) \leq e^\epsilon$.
This derivation is sufficient to show that $M$ enforces $\epsilon$-differential privacy, as the definition of differential privacy for this system requires $p(M(\vec{\response}) = \vec{a}) \leq e^\epsilon p(M(\vec{\response}') = \vec{a})$ for neighboring datasets $\vec{\response}, \vec{\response}'$ and all $\vec{a}$ in the range of $M$~\cite{DBLP:journals/fttcs/DworkR14}.

\begin{algorithm}[t]
	\caption{Binary Randomized Response with $n$ responses.}\label{fig:randomized-response}
	\begin{align*}
		&\stmtAsgn{\brrOutput}{0}\symSemi \\
		&\symFor~{i} \texttt{ in 1..}{n} ~\blockStart \\
		&\quad \stmtDist{\brrResponse_i}{\bernoulli(\brrResponseProb)}\symSemi \\
		&\quad \stmtDist{\brrRandomize_i}{\bernoulli\left(\brrRandomizeProb\right)}\symSemi \\
		&\quad \symIf~{\brrRandomize_i}~\blockStart \stmtAsgn{\brrOutput}{\brrOutput + \brrResponse_i} \blockEnd~\symElse~\blockStart \stmtAsgn{\brrOutput}{\brrOutput + \neg \brrResponse_i} \blockEnd \\
		&\blockEnd \\
		&\stmtAsgn{o}{o/n}
	\end{align*}
\end{algorithm}

\Cref{fig:randomized-response} implements this mechanism in a program that returns the average ratio of positive answers in a poll---this could be, e.g., the ratio of people who voted for a specific political party.
Variable $o$ is the average ratio of positive responses.
Since the mechanism ensures $\epsilon$-differential privacy for each answer, by the post-processing theorem~\cite{DBLP:journals/fttcs/DworkR14}, $o$ is also $\epsilon$-differentially private.
We are concerned with quantifying the amount of leakage for each variable $\response_i$ given that the attacker has access to $o$.
Variable $\epsilon$ models the level of privacy.
Values of $\epsilon$ close to 0 mean high privacy.
The larger the value of  $\epsilon$ the lower the level of privacy---up to $\epsilon = \infty$ meaning no privacy protection.
The value of $p$ defines attacker prior knowledge.
For instance, $p=1/2$ models non-informative attackers without background knowledge; i.e., the attacker considers positive and negative answers equally likely.
Different values of $p$ model attackers with different types of background knowledge---e.g., attackers with access to information indicating that a positive or negative answer is more likely.
Below we derive analytical solutions for leakage metrics which depend on $\epsilon$ and $p$, thus allowing us to analyze leakage for different types of attacker background knowledge ($p$) and privacy levels ($\epsilon$).
Since the overall output is $\epsilon$-differentially private and each loop iteration is independent, we focus only on the loop body in the analysis below.
Let $\stmt$ denote the loop body of the above program.

\paragraph*{Entropy of responses.}
As an introductory application of \Cref{proposition:wp-entropy}, we first calculate the entropy of the distribution of responses, $\entropy(\trv{\brrResponse_i})$:
\begin{align*}
	\entropy(\trv{\brrResponse_i}) \quad&=\quad \cwp{\stmt}(\alog\condforce{\cwp{\stmt}(\iverson{\aux{\brrResponse_i} = \brrResponse_i}, \expOne)\substBy{\aux{\brrResponse_i}}{\brrResponse_i}}, \expOne)~.
\end{align*}
We proceed with the calculation inside out and start with the inner $\symCwp$ expression.
This calculation is decomposed two further $\symWp$ calculations.
\begin{align*}
	\cwp{\stmt}(\iverson{\aux{\brrResponse_i} = \brrResponse_i}, \expOne) \morespace{=} (\wp{\stmt}(\iverson{\aux{\brrResponse_i} = \brrResponse_i}), \wp{\stmt}(\expOne))~.
\end{align*}
Notice that $\stmt$ does not contain loops itself, nor does it contain $\symObserve$ statements.
By standard rules of $\symWp$, we have $\wp{\stmt}(\expOne) = \expOne$.
The calculation for $\wp{\stmt}(\iverson{\aux{\brrResponse_i} = \brrResponse_i})$ proceeds bottom-up through the program, however only the first line in $\stmt$ has an effect on the expectation $\iverson{\aux{\brrResponse_i} = \brrResponse_i}$.
Following the definition of $\symWp$ (\Cref{table:wp-semantics}), we obtain
\begin{align*}
	\wp{\stmt}(\iverson{\aux{\brrResponse_i} = \brrResponse_i}) \morespace{=} p \cdot \iverson{\aux{\brrResponse_i} = \symTrue} + (1-p) \cdot \iverson{\aux{\brrResponse_i} = \symFalse}~.
\end{align*}
After substitution of $\aux{\brrResponse_i}$ by $\brrResponse_i$ and evaluating the conditional expectation, we get:
\begin{align*}
	\condforce{\cwp{\stmt}(\iverson{\aux{\brrResponse_i} = \brrResponse_i}, \expOne)\substBy{\aux{\brrResponse_i}}{\brrResponse_i}} \morespace{=} p \cdot \iverson{{\brrResponse_i} = \symTrue} + (1-p) \cdot \iverson{{\brrResponse_i} = \symFalse}~.
\end{align*}
A similar calculation and simplifications for the outer $\symCwp$ expression yield the final expectation for the entropy of the distribution of responses:
\begin{align*}
	\entropy(\trv{\brrResponse_i})
		&= p \cdot \alog(p \cdot \iverson{\symTrue = \symTrue} + (1-p) \cdot \iverson{\symFalse = \symTrue}) \\
		&\quad+ (1-p) \cdot \alog(p \cdot \iverson{\symTrue = \symFalse} + (1-p) \cdot \iverson{\symFalse = \symFalse}) \\
		&= p \cdot \alog(p) + (1-p) \cdot \alog(1-p)~.
\end{align*}
As expected, this is the well-known expression for the entropy of the Bernoulli distribution with parameter $p$.
The entropy $\entropy(\trv{\brrResponse_i})$ is maximal for $p = 0.5$ and zero for the extremal values $p = 0$ and $p = 1$.
\paragraph*{Entropy of outputs.}
As a more involved application of \Cref{proposition:wp-entropy}, we calculate the entropy of the distribution of outputs, $\entropy(\trv{\brrOutput_i})$.
We follow the same steps as above and start with $\wp{\stmt}(\iverson{\aux{\brrOutput_i} = \brrOutput_i})$.
The calculation is shown in \Cref{fig:output-entropy-wp-calculation}.
\begin{figure}[t]
	\begin{align*}
	&\wpcomp{\frac{1}{e^\epsilon + 1} \cdot \left(\iverson{\aux{\brrOutput_i} = \symTrue} \cdot \left(p \cdot e^\epsilon - p + 1\right) + \iverson{\aux{\brrOutput_i} = \symFalse} \cdot \left(-p \cdot e^\epsilon + p + e\right) \right) } \\
	&\wpcomp{
		\begin{aligned}[t]
			&p \cdot \left(\frac{e^\epsilon}{e^\epsilon + 1} \cdot \iverson{\aux{\brrOutput_i} = \symTrue} + \frac{1}{e^\epsilon + 1} \cdot \iverson{\neg \aux{\brrOutput_i} = \symFalse} \right) \\
			&\quad + \left(1-p\right) \cdot \left(\frac{e^\epsilon}{e^\epsilon + 1} \cdot \iverson{\aux{\brrOutput_i} = \symFalse} + \frac{1}{e^\epsilon + 1} \cdot \iverson{\neg \aux{\brrOutput_i} = \symTrue} \right)
		\end{aligned}
	} \\
	&\stmtDist{\brrResponse_i}{\bernoulli\left(\brrResponseProb\right)}\symSemi \\
	&\wpcomp{ \frac{e^\epsilon}{e^\epsilon + 1} \cdot \iverson{\aux{\brrOutput_i} = \brrResponse_i} + \frac{1}{e^\epsilon + 1} \cdot \iverson{\aux{\brrOutput_i} = \neg \brrResponse_i} } \\
	&\stmtDist{\brrRandomize_i}{\bernoulli\left(\brrRandomizeProb\right)}\symSemi \\
	&\wpcomp{ \iverson{\brrRandomize_i} \cdot \iverson{\aux{\brrOutput_i} = \brrResponse_i} + \iverson{\neg \brrRandomize_i} \cdot \iverson{\aux{\brrOutput_i} = \neg \brrResponse_i} } \\
	&\symIf~{\brrRandomize_i}~\blockStart \\
	&\quad \wpcomp{ \iverson{\aux{\brrOutput_i} = \brrResponse_i} } \\
	&\quad \stmtAsgn{\brrOutput_i}{\brrResponse_i} \\
	&\quad \wpcomp{ \iverson{\aux{\brrOutput_i} = \brrOutput_i} } \\
	&\blockEnd~\symElse~\blockStart \\
	&\quad \wpcomp{ \iverson{\aux{\brrOutput_i} = \neg \brrResponse_i} } \\
	&\quad \stmtAsgn{\brrOutput_i}{\neg \brrResponse_i} \\
	&\quad \wpcomp{ \iverson{\aux{\brrOutput_i} = \brrOutput_i} } \\
	&\blockEnd \\
	&\wpcomp{ \iverson{\aux{\brrOutput_i} = \brrOutput_i} }
	\end{align*}

	\caption{Calculation of $\wp{\stmt}(\iverson{\aux{\brrOutput_i} = \brrOutput_i})$ to compute the entropy of outputs. The calculation proceeds bottom-up, with the final result being in the first line.}
	\label{fig:output-entropy-wp-calculation}
\end{figure}
We substitute the ghost variable $\underline{\brrOutput_i}$ by the program variable $\brrOutput_i$ and obtain:
\[
	\frac{1}{e^\epsilon + 1} \cdot \left( \iverson{\brrOutput_i = \symTrue} \cdot (p \cdot e^\epsilon - p + 1) + \iverson{\brrOutput_i = \symFalse} \cdot (-p \cdot e^\epsilon + p + e) \right)~.
\]
We apply the negated logarithm to this result and simplify:
\begin{align}
	&\begin{array}{ll}
		&\iverson{\brrOutput_i = \symTrue} \cdot \alog(p \cdot e^\epsilon - p + 1) \\
		&\quad + \iverson{\brrOutput_i = \symFalse} \cdot \alog(-p \cdot e^\epsilon + p + e) + \palog(e^\epsilon + 1)~.
	\end{array} \label{eq:wp-log-inner}
\end{align}
The next step is to compute the conditional expected value of the above expression again.
To avoid a lengthy calculation as in \Cref{fig:output-entropy-wp-calculation}, we instead use the fact that our expression for $\wp{\stmt}(\iverson{\aux{\brrOutput_i} = \brrOutput_i})$ describes the probability distribution function of results.
To make use this fact within the expectation transformer framework, we use several algebraic properties of the $\symWp$ transformer.
First, we notice that since $\aux{\brrOutput_i}$ is a ghost variable and by definition not modified in $\stmt$, we have for all expectations $\expA \in \Expectations$ and values $v$ for $\aux{\brrOutput_i}$:
\[
	\wp{\stmt}\left(\expA\substBy{\aux{\brrOutput_i}}{v}\right) = \wp{\stmt}\left(\expA\right)\substBy{\aux{\brrOutput_i}}{v}~.
\]
Therefore, we can conclude from the initial computation only using properties of the $\symWp$ transformer:
\begin{align*}
	\wp{\stmt}\left(\iverson{\symTrue = \brrOutput_i}\right) &= \wp{\stmt}\left(\iverson{\aux{\brrOutput_i} = \brrOutput_i}\right)\substBy{\aux{\brrOutput_i}}{\symTrue} = \frac{p \cdot e^\epsilon - p + 1}{e^\epsilon + 1} \\
	\wp{\stmt}\left(\iverson{\symFalse = \brrOutput_i}\right) &= \wp{\stmt}\left(\iverson{\aux{\brrOutput_i} = \brrOutput_i}\right)\substBy{\aux{\brrOutput_i}}{\symFalse} = \frac{-p \cdot e^\epsilon + p + e}{e^\epsilon + 1}
\end{align*}
To compute the (conditional) expected value of our intermediate result in \Cref{eq:wp-log-inner} and the above, we use linearity of expected values.
For the weakest pre-expectation transformer, this means that we can apply a quantitative \emph{frame rule}~\cite{DBLP:journals/pacmpl/BatzKKMN19} that allows to separate parts of the expectation from the computation that are not modified in $\stmt$.
Applying the rule of~\cite{DBLP:journals/pacmpl/BatzKKMN19} to our setting, we get:
\begin{align*}
	&\wp{\stmt}\left(\alog\left(\wp{\stmt}\left(\iverson{\aux{\brrOutput_i} = \brrOutput_i}\right)\substBy{\aux{\brrOutput_i}}{\brrOutput_i}\right)\right)\\
	&= \frac{p \cdot e^\epsilon - p + 1}{e^\epsilon + 1} \cdot \alog\left(\wp{\stmt}\left(\iverson{\symTrue = \brrOutput_i}\right)\right) \\
	&\quad+ \frac{-p \cdot e^\epsilon + p + e}{e^\epsilon + 1} \cdot \alog\left(\wp{\stmt}\left(\iverson{\symFalse = \brrOutput_i}\right)\right) + \palog\left(e^\epsilon + 1\right) \\
	&= \frac{p \cdot e^\epsilon - p + 1}{e^\epsilon + 1} \cdot\alog\left(p \cdot e^\epsilon - p + 1\right) \\
	&\quad+ \frac{-p \cdot e^\epsilon + p + e}{e^\epsilon + 1} \cdot \alog\left(-p \cdot e^\epsilon + p + e\right) + \palog\left(e^\epsilon + 1\right)~. \\
\end{align*}
\Cref{fig:ent-outputs} shows the values of $\entropy(\trv{\brrOutput_i})$ in dependence on $p$ and $\epsilon$.
For example, let $\epsilon = 0$ and $p = 0.5$.
Then, ${e^\epsilon}/(e^\epsilon + 1) = 0.5$.
The result of the evaluation is $1$, which is the maximal value for all $\epsilon$.
More generally, $\epsilon = 0$ yields maximal entropy for arbitrary $p$.
This confirms that $\epsilon = 0$ will provide the largest distortion.
\begin{figure}[t]%
	\centering
	\begin{minipage}{.5\textwidth}
		\centering
		\includegraphics[width=\textwidth]{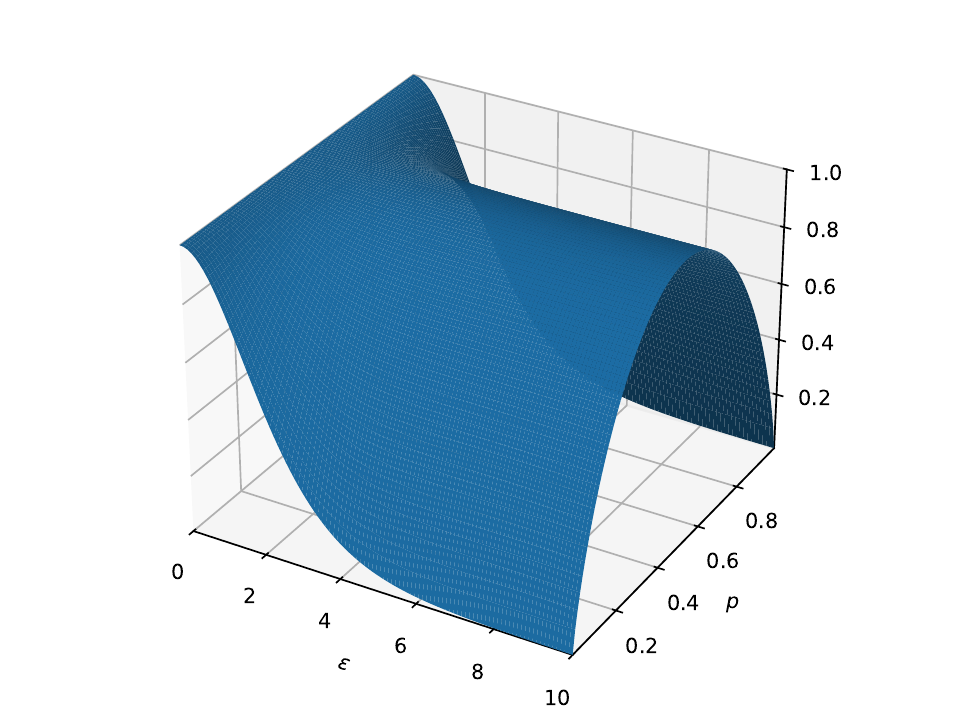}
		\captionof{figure}{Entropy of outputs $\entropy(\trv{\brrOutput_i})$.}
		\label{fig:ent-outputs}
	\end{minipage}%
	\begin{minipage}{.5\textwidth}
		\centering
		\includegraphics[width=\textwidth]{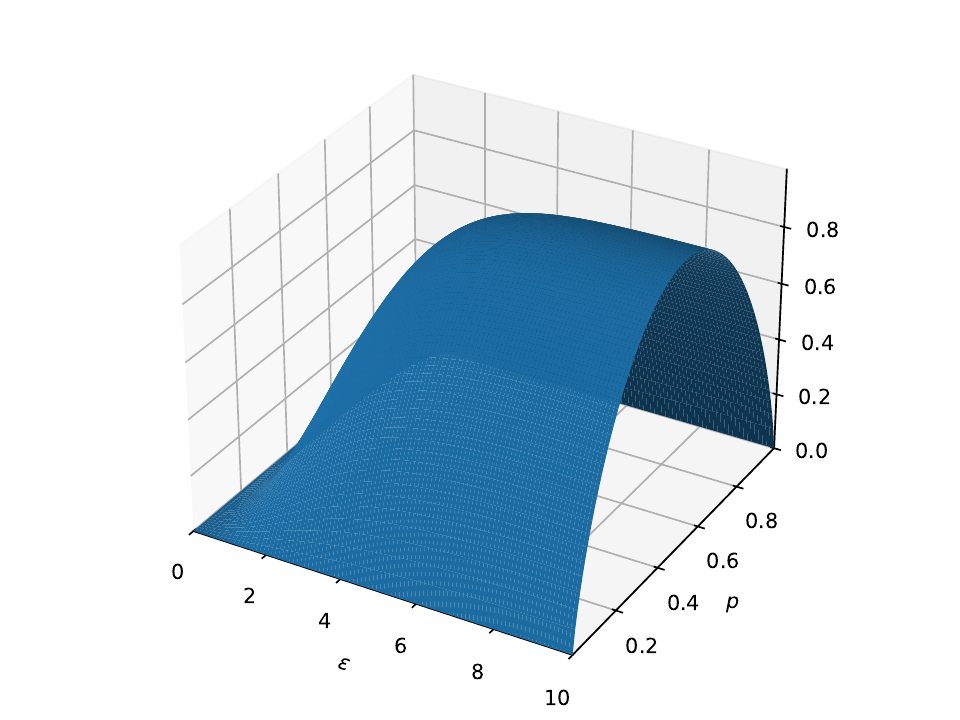}
		\captionof{figure}{Mutual information $\mi(\trv{\brrResponse_i}; \trv{\brrOutput_i})$\\ between secret response and output.}
		\label{fig:mi-response-output}
	\end{minipage}
\end{figure}
\paragraph{Conditional entropy.}
For the conditional entropy, we want to consider the case of $\entropy(\trv{\brrResponse_i} \mid \trv{\brrOutput_i})$, i.e. the conditional entropy of the value of the secret $\brrResponse_i$ given the output $\brrOutput_i$.
By \Cref{proposition:wp-cond-entropy}, we can calculate $\entropy(\trv{\brrResponse_i} \mid \trv{\brrOutput_i})$ as follows:
\begin{align*}
	&\cwp{\stmt}\left(\alog\left\condforceStart\cwp{\stmt}\left(\iverson{\aux{\brrOutput_i} = \brrOutput_i \land \aux{\brrResponse_i} = \brrResponse_i}, \iverson{\aux{\brrOutput_i} = \brrOutput_i}\right)\right\condforceEnd\substBy{\aux{\brrOutput_i}}{\brrOutput_i}\substBy{\aux{\brrResponse_i}}{\brrResponse_i}, \expOne \right)~.
\end{align*}
To compute the first parameter of the inner $\symCwp$ expression, one computes $\wp{\stmt}(\iverson{\aux{\brrOutput_i} = \brrOutput_i \land \aux{\brrResponse_i} = \brrResponse_i})$ which gives the joint probability distribution represented as an expectation.
For the second parameter of the inner $\symCwp$ expression, one again needs $\wp{\stmt}(\iverson{\aux{\brrOutput_i} = \brrOutput_i})$ (as calculated in \Cref{fig:output-entropy-wp-calculation}).
Because of its size, we omit the resulting expression for arbitrary $\epsilon$ and instead consider the specific case of $\epsilon = 0$ as an example again.
For $\epsilon = 0$, we get $\entropy(\trv{\brrResponse_i} \mid \trv{\brrOutput_i}) = -p \cdot \palog(p) - (1-p) \cdot \palog(1-p)$, which is the entropy of the original response $\brrResponse_i$.
In general, $\entropy(\trv{\brrResponse_i} \mid \trv{\brrOutput_i})$ will increase as $p$ goes towards $0.5$ and decrease with increasing values for $\epsilon$.
This shows that with increasing values of $\epsilon$, the attacker will have less uncertainty about the secret $\brrResponse_i$ after observing the output $\brrOutput_i$.
\paragraph*{Mutual information.}
The difference of the output entropy $\entropy(\trv{\brrOutput_i})$ and conditional entropy $\entropy(\trv{\brrResponse_i} \mid \trv{\brrOutput_i})$ define the mutual information $\mi(\trv{\brrResponse_i}; \trv{\brrOutput_i})$ between the secret $\brrResponse_i$ and the output $\brrOutput_i$.
\Cref{fig:mi-response-output} shows the values of $\mi(\trv{\brrResponse_i}; \trv{\brrOutput_i})$ in dependence on $p$ and $\epsilon$.
For $\epsilon = 0$, the mutual information $\mi(\trv{\brrResponse_i}; \trv{\brrOutput_i})$ is equal to zero.
In this case, the output does not tell an attacker anything about the secret response.
For larger values of $\epsilon$, more information will leak from the imperfect randomization and the mutual information will increase.
Let $p = 0.5$.
Then, we obtain mutual information values of approximately $0.002$ for $\epsilon = 0.1$ and approximately $0.999$ for $\epsilon = 10$.

\subsection{Gaussian DP (continuous)}

\begin{algorithm}[t]
	\caption{Gaussian Differential Privacy}
	\label{alg:GDP}
	\begin{algorithmic}[1]
		\STATE $\nfem := 0;$
		\FOR{$i$ \textbf{in range($4$)}}
			\STATE $\female_i \sim \bernoulli(0.75);$
			\STATE $\income_i \sim \normal(10,1);$
			\STATE $\nfem := \nfem + \female_i;$
		\ENDFOR
		\STATE \textbf{observe($nfem > 0$)};
		\STATE $\mu := 0;$
		\FOR{ $i$ \textbf{in range($4$)}}
			\IF{$\female_i = 1$}
			\STATE $\mu := \mu + \income_i;$
			\ELSE
			\STATE$\textbf{skip;}$
			\ENDIF
		\ENDFOR
		\IF{$\nfem=1$}
			\STATE $\nu := \normal(0, \var_1);$
		\ELSIF{$\nfem = 2$}
			\STATE $\nu := \normal(0, \var_2);$
			\STATE$\mu := 0.5\mu;$
		\ELSIF{$\nfem = 3$}
			\STATE $\nu := \normal(0, \var_3);$
			\STATE$\mu := 0.333\mu;$
		\ELSIF{$\nfem = 4$}
			\STATE $\nu := \normal(0, \var_4);$
			\STATE $\mu := 0.25\mu;$
		\ENDIF
		\STATE $\out := \mu + \nu;$
		\STATE \textbf{observe($\out = v$)}
	\end{algorithmic}
\end{algorithm}

The case study below uses the Gaussian mechanism to ensure differential privacy~\cite{DBLP:journals/fttcs/DworkR14} when releasing the average income of the females in the dataset.
The model under analysis is reported in \Cref{alg:GDP}.
We quantify the amount of information leaked in variables $\income_i$ with $i$ such that $\female_i$ equals true (as $\female_i$ indicates whether individual $i$ is female and $\income_i$ indicates his/her income).

To better show how the computation of the program semantics works, we consider a toy dataset with only four entries.
An individual in the dataset is a female with probability $0.75$ (line 3) and has an income distributed normally with mean $10$ and unitary variance (line 4).
These distributions model the attacker's prior knowledge about the gender and income of the individuals in the dataset, respectively.
The variable $\nfem$ counts the total number of females present in the dataset, and we constrain it to be strictly positive (line 7). Variable $\mu$ is the aggregate information that we want to safely release. It stores the mean female income, computed in the for loop in lines 9-15 and normalized in the if statement in lines 16-27.
In order to disclose $\mu$ without leaking information on single incomes, the program injects Gaussian noise $\nu$ to $\mu$ (line 28).
Variable $\nu$ is assigned in lines 16-27, where the variance is computed as $\var = 2 \Delta^2 \log(1.25/\delta)/\epsilon^2$ where $\Delta = \frac{\max_\income - \min_\income}{\nfem}$ and $\delta = 0.1$.
The variance and mean of the noise distribution are set according to the Gaussian mechanism~\cite{DBLP:journals/fttcs/DworkR14}.
Variables $\Delta$, $\epsilon$ and $\delta$ are specific to the differential privacy mechanism.
Variable $\Delta$ is known as sensitivity and it captures how much the output $\mu$ can change by removing one record in the dataset.
Variable $\epsilon$ is the privacy level and it has the same meaning as in the randomized response case study~(\Cref{subsec:randomized_response}).
Variable $\delta$ is the error probability, and it is used to relax the definition of differential privacy.\footnote{A mechanism $M \colon \mathbb{D} \to \mathbb{O}$---for any domains of input datasets $\mathbb{D}$ and mechanism outputs $\mathbb{O}$---is $(\epsilon, \delta)$-differentially private if for datasets $D, D' \in \mathbb{D}$ differing in one record, it holds that $p(M(D) \in O) \leq e^\epsilon p(M(D') \in O) + \delta$ for all $O \in \mathbb{O}$~\cite{DBLP:journals/fttcs/DworkR14}.}
Finally, we consider both a constrained version of the program, where in line 29 we observe a particular value of $\out$ (we take $v = 9$), and an unconstrained version, in which line 29 is suppressed.
By instantiating the properties presented in~\Cref{subsec:continuous_gmm_semantics}, we perform a leakage analysis for different privacy levels ($\epsilon$).

To investigate information leakage, we consider two distinct scenarios: for large values of $\epsilon$ ($\epsilon = 100$), we show that the knowledge of $\out$ gives a potential attacker information about $\income_i$; for small values of $\epsilon$ ($\epsilon = 0.1$) instead, we show that the leakage is greatly reduced, in accordance with the theory of $\epsilon$ differential privacy. We observe that \Cref{alg:GDP} satisfies the sufficient conditions of \Cref{prop:exact}, so the semantics computed by SOGA is the exact semantics.

To compute leakage metrics, we first need to compute the posterior distribution of the unconstrained program using the rules in Table \ref{tab:gaussian-semantics}. Let us give an outline of how the computation is performed. We suppose that at the beginning of the execution the distribution over all program variables is a Dirac delta centered in 0. Then, the first instruction does not modify the distribution. When entering the loop in line 2, $\female_0$ is assigned with a Bernoulli distribution. Applying the corresponding rule of \Cref{tab:gaussian-semantics}, one obtains that that after line 3 the joint over all program variables is a Gaussian mixture with two components: one with weight 0.75, $\female_0 = 1$, and one with weight 0.25 and $\female_0 = 0$. In both components, all other variables are set to zero. The next assignment modifies $\income_0$ in both components setting the mean to 10 and the variance to 1. Finally, the last instruction in the body loop increments $\nfem$ by $\female_0$. Since the update is performed in each component of the mixture, after the first iteration of the loop, the joint distribution is a mixture of two components: one with weight 0.75, $\female_0 = 1$ and $\nfem = 1$ and another component with weight 0.25, $\female_0 = 0$ and $\nfem = 0$. The following iterations will be computed in a similar way, and will increase the number of components up to 16 components. The observe statement in line 7 will have the effect of the discarding the component of the mixture where all $\female_i$ are set to zero and $\nfem=0$. The rest of the computation is carried out in a similar way.

As a result of the full computation, the joint is a mixture of 15 components. However, when considering the marginals, some components will have the same mean and variances, so can be lumped together as a single components. Using such marginals we compute the exact values of the metrics via numerical integration (with the Python library SciPy), while upper and lower bounds are derived analytically using \Cref{prop:Hbounds,prop:Lbounds}.

We report the results in \Cref{tab:GDP}. For $\entropy(\income_i)$ we only report the exact value since the marginal over $\income_i$ is a Gaussian and it can be computed exactly in closed form. To compute $\entropy(\income_i \mid \out = v)$, we consider the marginal over $\income_i$ in the constrained version of the program. The variable $\out$ is dependent on $\income_i$, therefore conditioning its value modifies the distribution over $\income_i$. Since in different components $\income_i$ and $\out$ have different covariances, the conditioning causes the distribution of $\income_i$ to become a GM. Moreover, depending on the value of $\epsilon$ the conditioning affects the distribution in a very different way: for $\epsilon = 100$, the variance of $\out$ is relatively small ($\sim 0.71$) and comparable to the covariance with $\income_i$ ($\sim 0.25$), so the observation affects significantly the distribution of $\income_i$ and consequently, its entropy; when $\epsilon = 0.1$ instead, the variance of $\out$ is much larger ($\sim 8333$) while the covariance is the same, so the conditioning only slightly affects the distribution of $\income_i$ and the value entropy remains the same up to the fifth decimal digit. For the same reason, when $\epsilon = 0.1$, the value of $\entropy(\income_i \mid \out)$ remains very close to the value of $\entropy(\income_i)$. The same argument explains the results for $\mi(\income_i; \out)$ and KL divergence: in both cases we have larger values for $\epsilon = 100$ and values close to 0 for $\epsilon = 0.1$, confirming that to minimize information leakage smaller values of $\epsilon$ have to be used.
\looseness -1

\begin{table}[t]
  \centering
  \renewcommand \tabcolsep {4pt}
	\begin{tabular}{
      @{}
      lcccccc
      @{}
    }
		\toprule
     & \multicolumn{3}{c}{$\epsilon = 100$} & \multicolumn{3}{c}{$\epsilon = 0.1$} \\[-1mm]
		\emph{Metrics} & \emph{Exact} & \emph{L.B.} & \emph{U.B.} & \emph{Exact} & \emph{L.B.} & \emph{U.B.} \\
		\midrule
		$\entropy(\income_i)$ & 1.418938 & - & - & 1.418938 & - & - \\
		$\entropy(\income_i \mid \out = v)$ & 1.386703 & 1.234257 & 2.249604 & 1.418931 & 1.265505 & 2.490682 \\
		$\entropy(\income_i \mid \out)$ & 1.370381 & -0.534517 & 3.085123 & 1.418933 & -2.593989 & 1.733485 \\
		$\mi(\income_i ; \out)$ & 0.048557 & -1.666184 & 1.953455 & 5.e-6 & -0.314547 & 4.012927 \\
		$ \kl(\income_i || \out)$ & 0.072827 & -1.221348 &  0.225273 & 3.e{-8} & -0.500000 & 0.153426 \\
		\bottomrule
	\end{tabular}

  \bigskip

	\caption{Exact metrics and analytical upper and lower bounds for them obtained with SOGA semantics.}
	\label{tab:GDP}
\end{table}


%% file: acknowledgements.tex
\subsubsection*{Acknowledgements}

This work was partially supported by the ERC Advanced Research Grant FRAPPANT (grant no.\ 787914) and by the European Union's Horizon 2020 research and innovation programme under the Marie Skłodowska-Curie grant agreement no.\ 101008233.

%% file: references.bib
@article{10.1145/3622870,
author = {Schr\"{o}er, Philipp and Batz, Kevin and Kaminski, Benjamin Lucien and Katoen, Joost-Pieter and Matheja, Christoph},
title = {A Deductive Verification Infrastructure for Probabilistic Programs},
year = {2023},
issue_date = {October 2023},
publisher = {Association for Computing Machinery},
address = {New York, NY, USA},
volume = {7},
number = {OOPSLA2},
url = {https://doi.org/10.1145/3622870},
doi = {10.1145/3622870},
journal = {Proc. ACM Program. Lang.},
month = {oct},
articleno = {294},
numpages = {31}
}

@article{RandoneBIT24,
  author	   = {Francesca Randone and Luca Bortolussi and Emilio
                  Incerto and Mirco Tribastone},
  title		   = {Inference of Probabilistic Programs with
                  Moment-Matching Gaussian Mixtures},
  journal	   = {Proc. {ACM} Program. Lang.},
  year		   = 2024,
  volume	   = 8,
  number	   = {{POPL}},
  pages		   = {1882-1912},
  OPTdoi	   = {10.1145/3632905},
  OPTurl	   = {https://doi.org/10.1145/3632905}
}

@InProceedings{PSI,
  author       = {Timon Gehr and Sasa Misailovic and Martin T. Vechev},
  title	       = {{PSI:} Exact Symbolic Inference for Probabilistic Programs},
  year	       = 2016,
  booktitle    = {{CAV}'16},
  OPTbooktitle = {Computer Aided Verification - 28th International Conference, {CAV} 2016, Toronto, ON, Canada, July 17-23, 2016, Proceedings, Part {I}},
  pages	       = {62-83},
  series       = {LNCS},
  volume       = {9779},
  OPTdoi	       = {10.1007/978-3-319-41528-4\_4},
  OPTurl	       = {https://doi.org/10.1007/978-3-319-41528-4\_4}
}

@inproceedings{gehr_psi_2020,
	OPTaddress = {London UK},
	title = {$\lambda${PSI}: exact inference for higher-order probabilistic programs},
	isbn = {978-1-4503-7613-6},
	OPTurl = {https://dl.acm.org/doi/10.1145/3385412.3386006},
	OPTdoi = {10.1145/3385412.3386006},
	booktitle = {PLDI'20},
	publisher = {ACM},
	author = {Gehr, Timon and Steffen, Samuel and Vechev, Martin},
	OPTmonth = jun,
	year = {2020},
	pages = {883--897},
}

@InProceedings{spire,
  author       = {Martin Kucera and Petar Tsankov and Timon Gehr and
                  Marco Guarnieri and Martin T. Vechev},
  title	       = {Synthesis of Probabilistic Privacy Enforcement},
  year	       = {2017},
  booktitle    = {{CCS}'17},
  OPTbooktitle    = {Proceedings of the 2017 {ACM} {SIGSAC} Conference on Computer and Communications Security, {CCS} 2017, Dallas, TX, USA, October 30 - November 03, 2017},
  pages	       = {391-408},
  publisher    = {{ACM}},
  OPTisbn      = {978-1-4503-4946-8},
  OPTdoi	       = {10.1145/3133956.3134079},
  OPTurl	       = {https://doi.org/10.1145/3133956.3134079}
}

@InProceedings{QUAIL,
  author       = {Fabrizio Biondi and Axel Legay and Louis{-}Marie
                  Traonouez and Andrzej Wasowski},
  title	       = {{QUAIL:} {A} Quantitative Security Analyzer for
                  Imperative Code},
  booktitle    = {{CAV}'13},
  year         = {2013},
  publisher    = {Springer Berlin Heidelberg},
  pages        = {702-707},
  isbn         = {978-3-642-39799-8}
}

@article{DBLP:journals/fttcs/DworkR14,
  author       = {Cynthia Dwork and
                  Aaron Roth},
  title        = {The Algorithmic Foundations of Differential Privacy},
  journal      = {Found. Trends Theor. Comput. Sci.},
  volume       = {9},
  number       = {3-4},
  pages        = {211--407},
  year         = {2014},
  url          = {https://doi.org/10.1561/0400000042},
  doi          = {10.1561/0400000042}
}

@inproceedings{li2006t,
  author       = {Ninghui Li and
                  Tiancheng Li and
                  Suresh Venkatasubramanian},
  title        = {t-Closeness: Privacy Beyond k-Anonymity and l-Diversity},
  booktitle    = {Proceedings of the 23rd International Conference on Data Engineering,
                  {ICDE} 2007, The Marmara Hotel, Istanbul, Turkey, April 15-20, 2007},
  pages        = {106--115},
  publisher    = {{IEEE} Computer Society},
  year         = {2007},
  url          = {https://doi.org/10.1109/ICDE.2007.367856},
  doi          = {10.1109/ICDE.2007.367856}
}

@article{sweeney2002k,
  title={k-anonymity: A model for protecting privacy},
  author={Sweeney, Latanya},
  journal={International journal of uncertainty, fuzziness and knowledge-based systems},
  volume={10},
  number={05},
  pages={557--570},
  year={2002},
  publisher={World Scientific}
}

@InProceedings{DBLP:conf/sp/MardzielAHC14,
  author	   = {Piotr Mardziel and M{\'{a}}rio S. Alvim and Michael
                  W. Hicks and Michael R. Clarkson},
  title		   = {Quantifying Information Flow for Dynamic Secrets},
  year		   = 2014,
  booktitle	   = {2014 {IEEE} Symposium on Security and Privacy, {SP}
                  2014, Berkeley, CA, USA, May 18-21, 2014},
  publisher    = {{IEEE} Computer Society},
  isbn         = {978-1-4799-4686-0},
  pages		   = {540-555},
  doi		   = {10.1109/SP.2014.41},
  url		   = {https://doi.org/10.1109/SP.2014.41}
}

@InProceedings{DBLP:conf/csfw/ClarksonMS05,
  author	   = {Michael R. Clarkson and Andrew C. Myers and Fred
                  B. Schneider},
  title		   = {Belief in Information Flow},
  year		   = 2005,
  booktitle	   = {18th {IEEE} Computer Security Foundations Workshop,
                  {(CSFW-18} 2005), 20-22 June 2005, Aix-en-Provence,
                  France},
  pages		   = {31-45},
  year         = {2005},
  publisher    = {{IEEE} Computer Society},
  isbn         = {0-7695-2340-4},
  doi		   = {10.1109/CSFW.2005.10}
}

@Article{DBLP:journals/mscs/ClarksonS15,
  author	   = {Michael R. Clarkson and Fred B. Schneider},
  title		   = {Quantification of integrity},
  journal	   = {Math. Struct. Comput. Sci.},
  year		   = 2015,
  volume	   = 25,
  number	   = 2,
  pages		   = {207-258},
  doi		   = {10.1017/S0960129513000595},
  url		   = {https://doi.org/10.1017/S0960129513000595}
}

@Book{qifbook.2020,
  author    = 	 {M\'ario Alvim and Konstantinos Chatzikokolakis
                  and Annabelle McIver and Carroll Morgan and Catuscia
                  Palamidessi and Geoffrey Smith},
  title     = 	 {The Science of Quantitative Information Flow},
  publisher = 	 {Springer},
   address  =      "Cham",
  year      = 	 {2020},
}

@misc{RRPWqprgm24,
  author       = {Rasmus C. R{\o}nneberg and Francesca Randone and Ra{\'u}l Pardo and Andrzej W{\k{a}}sowski},
  OPTjournal   = {},
  title        = {Quantifying Privacy Risk with Gaussian Mixtures},
  OPTvolume    = {},
  OPTpages     = {},
  OPTpublisher = {},
  year         = {2024},
  OPTurl       = {},
  OPTdoi       = {},
  note         = {{\it Under submission}},
  OPTkeywords  = {journal}
}

@inproceedings{RPWeebiprq23,
  author    = {Rasmus C. R{\o}nneberg and Ra{\'u}l Pardo and Andrzej W{\k{a}}sowski},
  title     = {Exact and Efficient Bayesian Inference for Privacy Risk Quantification},
  booktitle = {Proceedings of 21st International Conference on Software Engineering and Formal Methods, {SEFM 2023}},
  series    = {Lecture Notes in Computer Science},
  volume    = {14323},
  pages  = {263-281},
  publisher = {Springer},
  year      = {2023},
  OPTurl    = {},
  OPTdoi    = {},
  OPTnote   = {}
}

@Article{DBLP:journals/popets/AlvimFMMN22,
  author	   = {M{\'{a}}rio S. Alvim and Natasha Fernandes and
                  Annabelle McIver and Carroll Morgan and Gabriel
                  Henrique Nunes},
  title		   = {Flexible and scalable privacy assessment for very
                  large datasets, with an application to official
                  governmental microdata},
  journal	   = {Proc. Priv. Enhancing Technol.},
  year		   = 2022,
  volume	   = 2022,
  number	   = 4,
  pages		   = {378-399},
  doi		   = {10.56553/POPETS-2022-0114},
  url		   = {https://doi.org/10.56553/popets-2022-0114}
}

@InProceedings{DBLP:conf/wpes/AlvimFMN23,
  author	   = {M{\'{a}}rio S. Alvim and Natasha Fernandes and
                  Annabelle McIver and Gabriel Henrique Nunes},
  title		   = {A Quantitative Information Flow Analysis of the
                  Topics {API}},
  year		   = 2023,
  booktitle	   = {Proceedings of the 22nd Workshop on Privacy in the
                  Electronic Society, {WPES} 2023, Copenhagen,
                  Denmark, 26 November 2023},
  pages		   = {123-127},
  publisher    = {{ACM}},
  doi		   = {10.1145/3603216.3624959},
  url		   = {https://doi.org/10.1145/3603216.3624959}
}

@InProceedings{DBLP:conf/sp/BackesKR09,
  author	   = {Michael Backes and Boris K{\"{o}}pf and Andrey
                  Rybalchenko},
  title		   = {Automatic Discovery and Quantification of
                  Information Leaks},
  year		   = 2009,
  booktitle	   = {30th {IEEE} Symposium on Security and Privacy {(SP}
                  2009), 17-20 May 2009, Oakland, California, {USA}},
  pages		   = {141-153},
  publisher    = {{IEEE} Computer Society},
  isbn         = {978-0-7695-3633-0},
  doi		   = {10.1109/SP.2009.18},
  url		   = {https://doi.org/10.1109/SP.2009.18}
}

@InProceedings{DBLP:conf/eurosp/BangRB18,
  author	   = {Lucas Bang and Nicol{\'{a}}s Rosner and Tevfik
                  Bultan},
  title		   = {Online Synthesis of Adaptive Side-Channel Attacks
                  Based On Noisy Observations},
  year		   = 2018,
  booktitle	   = {2018 {IEEE} European Symposium on Security and
                  Privacy, EuroS{\&}P 2018, London, United Kingdom,
                  April 24-26, 2018},
  pages		   = {307-322},
  publisher    = {{IEEE}},
  isbn         = {978-1-5386-4228-3},
  doi		   = {10.1109/EUROSP.2018.00029},
  url		   = {https://doi.org/10.1109/EuroSP.2018.00029}
}

@InProceedings{DBLP:conf/sp/Cherubin0P19,
  author	   = {Giovanni Cherubin and Konstantinos Chatzikokolakis
                  and Catuscia Palamidessi},
  title		   = {{F-BLEAU:} Fast Black-Box Leakage Estimation},
  year		   = 2019,
  booktitle	   = {2019 {IEEE} Symposium on Security and Privacy, {SP}
                  2019, San Francisco, CA, USA, May 19-23, 2019},
  pages		   = {835-852},
  publisher    = {{IEEE}},
  isbn         = {978-1-5386-6660-9},
  doi		   = {10.1109/SP.2019.00073},
  url		   = {https://doi.org/10.1109/SP.2019.00073}
}

@InProceedings{DBLP:conf/csfw/ChothiaG11,
  author	   = {Tom Chothia and Apratim Guha},
  title		   = {A Statistical Test for Information Leaks Using
                  Continuous Mutual Information},
  year		   = 2011,
  booktitle	   = {Proceedings of the 24th {IEEE} Computer Security
                  Foundations Symposium, {CSF} 2011, Cernay-la-Ville,
                  France, 27-29 June, 2011},
  pages		   = {177-190},
  publisher    = {{IEEE} Computer Society},
  isbn         = {978-1-61284-644-6},
  doi		   = {10.1109/CSF.2011.19},
  url		   = {https://doi.org/10.1109/CSF.2011.19}
}

@InProceedings{DBLP:conf/cav/ChothiaKN13,
  author	   = {Tom Chothia and Yusuke Kawamoto and Chris Novakovic},
  title		   = {A Tool for Estimating Information Leakage},
  year		   = 2013,
  booktitle	   = {Computer Aided Verification - 25th International
                  Conference, {CAV} 2013, Saint Petersburg, Russia,
                  July 13-19, 2013. Proceedings},
  pages		   = {690-695},
  series       = {Lecture Notes in Computer Science},
  publisher    = {Springer},
  volume       = {8044},
  isbn         = {978-3-642-39798-1},
  doi		   = {10.1007/978-3-642-39799-8\_47},
  url		   = {https://doi.org/10.1007/978-3-642-39799-8\_47}
}

@InProceedings{DBLP:conf/esorics/ChothiaKN14,
  author	   = {Tom Chothia and Yusuke Kawamoto and Chris Novakovic},
  title		   = {LeakWatch: Estimating Information Leakage from Java
                  Programs},
  year		   = 2014,
  booktitle	   = {Computer Security - {ESORICS} 2014 - 19th European
                  Symposium on Research in Computer Security, Wroclaw,
                  Poland, September 7-11, 2014. Proceedings, Part
                  {II}},
  pages		   = {219-236},
  series       = {Lecture Notes in Computer Science},
  volume       = {8713},
  publisher    = {Springer},
  doi		   = {10.1007/978-3-319-11212-1\_13},
  url		   = {https://doi.org/10.1007/978-3-319-11212-1\_13}
}

@Article{DBLP:journals/ijon/GrossoPPP23,
  author	   = {Ganesh Del Grosso and Georg Pichler and Catuscia
                  Palamidessi and Pablo Piantanida},
  title		   = {Bounding information leakage in machine learning},
  journal	   = {Neurocomputing},
  year		   = 2023,
  volume	   = 534,
  pages		   = {1-17},
  doi		   = {10.1016/J.NEUCOM.2023.02.058},
  url		   = {https://doi.org/10.1016/j.neucom.2023.02.058}
}

@InProceedings{DBLP:conf/ccs/0002CPP20,
  author	   = {Marco Romanelli and Konstantinos Chatzikokolakis and
                  Catuscia Palamidessi and Pablo Piantanida},
  title		   = {Estimating g-Leakage via Machine Learning},
  year		   = 2020,
  booktitle	   = {{CCS} '20: 2020 {ACM} {SIGSAC} Conference on
                  Computer and Communications Security, Virtual Event,
                  USA, November 9-13, 2020},
  pages		   = {697-716},
  publisher    = {{ACM}},
  isbn         = {978-1-4503-7089-9},
  doi		   = {10.1145/3372297.3423363},
  url		   = {https://doi.org/10.1145/3372297.3423363}
}

@InProceedings{DBLP:conf/fossacs/Smith09,
  author	   = {Geoffrey Smith},
  title		   = {On the Foundations of Quantitative Information Flow},
  year		   = 2009,
  booktitle	   = {Foundations of Software Science and Computational
                  Structures, 12th International Conference, {FOSSACS}
                  2009, Held as Part of the Joint European Conferences
                  on Theory and Practice of Software, {ETAPS} 2009,
                  York, UK, March 22-29, 2009. Proceedings},
  pages		   = {288-302},
  volume       = {5504},
  publisher    = {Springer},
  series       = {Lecture Notes in Computer Science},
  isbn         = {978-3-642-00595-4},
  doi		   = {10.1007/978-3-642-00596-1\_21},
  url		   = {https://doi.org/10.1007/978-3-642-00596-1\_21}
}

@InProceedings{DBLP:conf/cav/Tizpaz-NiariC019,
  author	   = {Saeid Tizpaz{-}Niari and Pavol Cern{\'{y}} and
                  Ashutosh Trivedi},
  title		   = {Quantitative Mitigation of Timing Side Channels},
  year		   = 2019,
  booktitle	   = {Computer Aided Verification - 31st International
                  Conference, {CAV} 2019, New York City, NY, USA, July
                  15-18, 2019, Proceedings, Part {I}},
  pages		   = {140-160},
  series       = {Lecture Notes in Computer Science},
  volume       = {11561},
  publisher    = {Springer},
  isbn         = {978-3-030-25539-8},
  doi		   = {10.1007/978-3-030-25540-4\_8},
  url		   = {https://doi.org/10.1007/978-3-030-25540-4\_8}
}

@misc{weisstein,
    author   = {Weisstein, Eric W.},
    title    = {Normal Product Distribution. {From MathWorld---A Wolfram Web Resource}},
    url      = {\url{https://mathworld.wolfram.com/NormalProductDistribution.html}}
}

@inproceedings{huber2008entropy,
  title={On entropy approximation for Gaussian mixture random vectors},
  author={Huber, Marco F and Bailey, Tim and Durrant-Whyte, Hugh and Hanebeck, Uwe D},
  booktitle={2008 IEEE International Conference on Multisensor Fusion and Integration for Intelligent Systems},
  pages={181--188},
  year={2008},
  organization={IEEE}
}

@inproceedings{hershey2007approximating,
  title={Approximating the Kullback Leibler divergence between Gaussian mixture models},
  author={Hershey, John R and Olsen, Peder A},
  booktitle={2007 IEEE International Conference on Acoustics, Speech and Signal Processing-ICASSP'07},
  volume={4},
  pages={IV--317},
  year={2007},
  organization={IEEE}
}

@inproceedings{kozen1979semantics,
	title={Semantics of probabilistic programs},
	author={Kozen, Dexter},
	booktitle={20th Annual Symposium on Foundations of Computer Science (sfcs 1979)},
	pages={101--114},
	year={1979},
	organization={IEEE}
}

@book{bishop2006pattern,
  title={Pattern Recognition and Machine Learning},
  author={Bishop, Christopher M and Nasrabadi, Nasser M},
  volume={4},
  number={4},
  year={2006},
  publisher={Springer}
}

@legal{hipaa1996,
  author       = {{U.S. Congress}},
  title        = {{Health Insurance Portability and Accountability Act of 1996}},
  howpublished = {Public Law 104-191},
  year         = {1996},
}

@legal{gdpr,
  date       = {2016-05-04},
  location   = {OJ L 119, 4.5.2016, p. 1--88},
  title      = {Regulation ({EU}) 2016/679 of the {European} {Parliament} and of the {Council}},
  url        = {https://data.europa.eu/eli/reg/2016/679/oj},
  titleaddon = {of 27 {April} 2016 on the protection of natural persons with regard to the processing of personal data and on the free movement of such data, and repealing {Directive} 95/46/{EC} ({General} {Data} {Protection} {Regulation})},
  author     = {{European Parliament} and {Council of the European Union}},
  urldate    = {2023-04-13},
}

@legal{noauthor_cybersecurity_2016,
	title = {Cybersecurity {Law} of the {People}'s {Republic} of {China}},
	url = {http://www.lawinfochina.com/Display.aspx?LookType=3\&Lib=law\&Id=22826\&SearchKeyword=\&SearchCKeyword=\&paycode=},
	urldate = {2023-05-01},
	year = {2016},
}

@book{ShannonWeaver49,
  address = {Urbana, IL},
  author = {Shannon, Claude E. and Weaver, Warren},
  description = {1998 Reprint of the First Paperback Edition 1963},
  isbn = {978-0-252-72548-7},
  publisher = {University of Illinois Press},
  title = {The Mathematical Theory of Communication},
  username = {flint63},
  year = 1949
}

@Article{HyLeak,
  author       = {Fabrizio Biondi and Yusuke Kawamoto and Axel Legay
                  and Louis{-}Marie Traonouez},
  title	       = {Hybrid statistical estimation of mutual information
                  and its application to information flow},
  year	       = 2019,
  volume       = 31,
  number       = 2,
  pages	       = {165-206},
  OPTdoi       = {10.1007/s00165-018-0469-z},
  OPTurl       = {https://doi.org/10.1007/s00165-018-0469-z},
  journal      = {Formal Aspects Comput.}
}

@InProceedings{romanelli.leaves.2020,
  author       = {Marco Romanelli and Konstantinos Chatzikokolakis and
                  Catuscia Palamidessi and Pablo Piantanida},
  title	       = {Estimating g-Leakage via Machine Learning},
  year	       = 2020,
  booktitle    = {{CCS}'20},
  OPTbooktitle = {{CCS} '20: 2020 {ACM} {SIGSAC} Conference on Computer and Communications Security, Virtual Event, USA, November 9-13, 2020},
  publisher    = {{ACM}},
  OPTisbn      = {978-1-4503-7089-9},
  OPTdoi       = {10.1145/3372297.3423363},
  OPTurl       = {https://doi.org/10.1145/3372297.3423363},
}

@InProceedings{privug,
  author    = {Pardo, Ra{\'u}l and
               Rafnsson, Willard and
               Probst, Christian W. and
               Wasowski, Andrzej},
  title     = {Privug: Using Probabilistic Programming for Quantifying Leakage in Privacy Risk Analysis},
  booktitle = {ESORICS'21},
  series    = {LNCS},
  volume    = {12973},
  year      = {2021},
  publisher = {Springer},
}

@book{mcmc,
	title = {Monte {Carlo} {Statistical} {Methods}},
	publisher = {Springer},
	author = {Christian P. Robert and {George Casella}},
	year = {2004},
}

@phdthesis{kaminski_phd,
  author       = {Benjamin Lucien Kaminski},
  title        = {Advanced weakest precondition calculi for probabilistic programs},
  school       = {{RWTH} Aachen University, Germany},
  year         = {2019},
  url          = {http://publications.rwth-aachen.de/record/755408},
  urn          = {urn:nbn:de:101:1-2019072507004473397003},
  bibsource    = {dblp computer science bibliography, https://dblp.org}
}

@book{DBLP:series/mcs/McIverM05,
  author       = {Annabelle McIver and
                  Carroll Morgan},
  title        = {Abstraction, Refinement and Proof for Probabilistic Systems},
  series       = {Monographs in Computer Science},
  publisher    = {Springer},
  year         = {2005},
  url          = {https://doi.org/10.1007/b138392},
  doi          = {10.1007/B138392},
  isbn         = {978-0-387-40115-7},
  bibsource    = {dblp computer science bibliography, https://dblp.org}
}

@inproceedings{DBLP:conf/cav/ChenKKW22,
  author       = {Mingshuai Chen and
                  Joost{-}Pieter Katoen and
                  Lutz Klinkenberg and
                  Tobias Winkler},
  editor       = {Sharon Shoham and
                  Yakir Vizel},
  title        = {Does a Program Yield the Right Distribution? - Verifying Probabilistic
                  Programs via Generating Functions},
  booktitle    = {Computer Aided Verification - 34th International Conference, {CAV}
                  2022, Haifa, Israel, August 7-10, 2022, Proceedings, Part {I}},
  series       = {Lecture Notes in Computer Science},
  volume       = {13371},
  pages        = {79--101},
  publisher    = {Springer},
  year         = {2022},
  url          = {https://doi.org/10.1007/978-3-031-13185-1\_5},
  doi          = {10.1007/978-3-031-13185-1\_5},
  bibsource    = {dblp computer science bibliography, https://dblp.org}
}

@inproceedings{DBLP:journals/entcs/0001KKOGM15,
  author       = {Nils Jansen and
                  Benjamin Lucien Kaminski and
                  Joost{-}Pieter Katoen and
                  Federico Olmedo and
                  Friedrich Gretz and
                  Annabelle McIver},
  title        = {Conditioning in Probabilistic Programming},
  booktitle    = {{MFPS}},
  series       = {Electronic Notes in Theoretical Computer Science},
  volume       = {319},
  pages        = {199--216},
  publisher    = {Elsevier},
  year         = {2015}
}

@article{DBLP:journals/pacmpl/BatzKKMN19,
  author       = {Kevin Batz and
                  Benjamin Lucien Kaminski and
                  Joost{-}Pieter Katoen and
                  Christoph Matheja and
                  Thomas Noll},
  title        = {Quantitative separation logic: a logic for reasoning about probabilistic
                  pointer programs},
  journal      = {Proc. {ACM} Program. Lang.},
  volume       = {3},
  number       = {{POPL}},
  pages        = {34:1--34:29},
  year         = {2019}
}
